\documentclass[11pt]{amsart}

\makeatletter
 \def\@textbottom{\vskip \z@ \@plus 1pt}
 \let\@texttop\relax
\makeatother
\usepackage{amsmath, amssymb, amsbsy, amsfonts, amsthm, latexsym, amsopn, amstext, amsxtra, euscript, amscd, color, mathrsfs}
\usepackage[normalem]{ulem}
\usepackage{soul}

\makeatletter
\@namedef{subjclassname@2020}{%
  \textup{2020} Mathematics Subject Classification}
\makeatother

\PassOptionsToPackage{hyphens}{url}\usepackage{hyperref}
 
 \usepackage[capbesideposition=outside,capbesidesep=quad]{floatrow}

\restylefloat{table}
\restylefloat{table}
            
\usepackage{multirow,caption}
\usepackage{tablefootnote}         
\usepackage{amscd}
\usepackage{color,enumerate}
\newcommand{\RNum}[1]{\lowercase\expandafter{\romannumeral #1\relax}}

\usepackage[colorinlistoftodos,prependcaption,textsize=tiny]{todonotes}

\newtheorem{thm}{Theorem}[section]
\newtheorem{lem}[thm]{Lemma}
\newtheorem{cor}[thm]{Corollary}

\newtheorem{exmp}[thm]{Example}

\newtheorem{rem}[thm]{Remark}

\newtheorem{rmk}[thm]{Remark}

\newtheorem{thm-con}[thm]{Theorem-Conjecture}
\numberwithin{equation}{section}

\theoremstyle{definition}
\newtheorem{defn}[thm]{Definition}

\def\Im{{\rm Im}}

\newcommand{\cB}{\mathcal B}

\def\cA{{\mathcal A}}
\def\cC{{\mathcal C}}
\def\cD{{\mathcal D}}
\newcommand{\F}{\mathbb F} 
\DeclareMathOperator{\LB}{LBCT}
\DeclareMathOperator{\UB}{UBCT}
\DeclareMathOperator{\EB}{EBCT}
\DeclareMathOperator{\DB}{DBCT}
\DeclareMathOperator{\DDT}{DDT}

\DeclareMathOperator{\FB}{FBCT}
\DeclareMathOperator{\DD}{DDDT}

\newcommand*\colvec[3][]{
\begin{pmatrix}
\ifx\relax#1\relax\else#1\\\fi#2\\#3
\end{pmatrix}
}
\def\Tr{{\rm Tr}}

\def\Trsn{{\rm Tr}_s^n}

\begin{document}

\title[The revised boomerang connectivity tables]{The revised boomerang connectivity tables and their connection to the difference distribution table}

 \author[K. Garg]{Kirpa Garg}
  \address{Department of Mathematics, Indian Institute of Technology Jammu, Jammu 181221, India}
  \email{kirpa.garg@gmail.com}
  
  \author[S. U. Hasan]{Sartaj Ul Hasan}
  \address{Department of Mathematics, Indian Institute of Technology Jammu, Jammu 181221, India}
  \email{sartaj.hasan@iitjammu.ac.in}
  
  \author[C. Riera]{Constanza Riera}
  \address{Department of Computer Science, Electrical Engineering and Mathematical Sciences, Western Norway University of Applied Sciences, 5020 Bergen, Norway}
  \email{csr@hvl.no}
  
   \author[P.~St\u anic\u a]{Pantelimon~St\u anic\u a$^*$}
   \address{Applied Mathematics Department, Naval Postgraduate School, Monterey, CA 93943, USA}
  \email{pstanica@nps.edu}

 \thanks{The work of K. Garg is supported by the University Grants Commission (UGC), Government of India. 
 The work of S. U. Hasan is partially supported by Core Research Grant CRG/2022/005418 from the Science and Engineering Research Board, Government of India. The work of P. St\u anic\u a (corresponding author)  is partially supported by a grant from the NPS Foundation.}

\keywords{Finite fields, permutation polynomials, almost perfect non-linear functions, upper boomerang connectivity table, lower boomerang connectivity table, extended boomerang connectivity table, double boomerang connectivity table.}

\subjclass[2020]{12E20, 11T06, 94A60}

\begin{abstract} 
It is well-known that functions over finite fields play a crucial role in designing substitution boxes (S-boxes) in modern block ciphers. In order to analyze the security of an S-box, recently, three new tables have been introduced: the  Extended Boomerang Connectivity Table (EBCT), the Lower Boomerang Connectivity Table (LBCT), and the Upper Boomerang Connectivity Table (UBCT). In fact, these tables offer improved methods over the usual Boomerang Connectivity Table (BCT) for analyzing the security of S-boxes against boomerang-style attacks. Here, we put in context these new EBCT, LBCT, and UBCT concepts by connecting them to the DDT for a differentially $\delta$-uniform function and also determine the EBCT, LBCT, and UBCT entries of three classes of differentially $4$-uniform power permutations, namely, Gold, Kasami and Bracken-Leander. We also determine the Double Boomerang Connectivity Table (DBCT) entries of the Gold function. As byproducts of our approach, we obtain some previously published results quite easily.
\end{abstract}
\maketitle

\section{Introduction}
In symmetric cryptography, vectorial Boolean functions over finite fields play a significant role, particularly in designing block ciphers' S-boxes. There are various kinds of cryptanalytic attacks possible on these S-boxes. One well-known attack on S-boxes is the differential attack, introduced by Biham and Shamir~\cite{BS}. This attack exploits the non-uniformity in the distribution of output differences corresponding to a given input difference. To quantify the degree of resistance of S-Boxes against differential attacks, Nyberg~\cite{Nyberg} introduced the notion of Difference Distribution Table (DDT) and differential uniformity.

Let $n$ be a positive integer. We denote by $\F_{2^n}$ the finite field with $2^n$ elements, by $\F_{2^n}^{*}$ the multiplicative cyclic group of non-zero elements of $\F_{2^n}$ and by $\F_{2^n}[X]$ the ring of polynomials in one variable $X$ with coefficients in $\F_{2^n}$. For a function $F : \F_{2^n}\to \F_{2^n}$, and for any $a \in \F_{2^n}$, the derivative of $F$ in the direction $a$ is defined as $D_F(X,a) := F(X+a)+F(X)$ for all $X\in \F_{2^n}$. For any $a, b \in \F_{2^n}$, the Difference Distribution Table (DDT) entry $\DDT_F(a,b)$ at point $(a, b)$ is the number of solutions $X\in \F_{2^n}$ of the equation $D_F(X,a) = b$. Further, the differential uniformity of $F$, denoted by $\Delta_F$, is given by $\Delta_{F} := \max\left\{\DDT_{F}(a, b) \bigm|  a\in \F_{2^n}^*, b \in \F_{2^n} \right\}.$ We call a function $F$ an almost perfect nonlinear (APN) function if $\Delta_{F} = 2$. Notice that $\Delta_F$ is always even, because if $X$ is a solution of $D_F(X,a) = b$ then $X+a$ is also a solution of $D_F(X,a) = b$. Functions with low differential uniformity are more resistant to differential attacks, making them desirable for cryptographic applications.

However, it turns out that low differential uniformity is not sufficient to counter some differentials connected attacks. In 1999, Wagner~\cite{Wagner} introduced a new attack on block ciphers, called the boomerang attack, which can be seen as an extension of the differential attack. It allows the cryptanalyst to use two unrelated differential characteristics to attack the same cipher by using one differential to defeat the first half and another to defeat the second half of the cipher. The theoretical underpinning of this attack was  considered by Cid et al.~\cite{Cid}, who introduced the notion of the Boomerang Connectivity Table (BCT) that can be used to more accurately evaluate the probability of generating a right quartet in boomerang-style attacks such that the boomerang incompatibility, ladder switch and S-box switch can easily be detected. For effectively computing the entries in the BCT, Li et al.~\cite{LSL} proposed an equivalent formulation as described below, which avoids computing the inverse of the S-box, and consequently, it can be defined even for non-permutation. For any $a,b \in \F_{2^n}$, the BCT entry at $(a,b)\in \F_{2^n}\times \F_{2^n}$, denoted as $ \cB_F(a,b)$, is the number of solutions in $\mathbb{F}_{2^n}\times \mathbb{F}_{2^n}$ of the following system
\begin{equation*}
  \begin{cases} 
  F(X)+F(Y)=b \\
  F(X+a)+F(Y+a)=b.
  \end{cases}
\end{equation*} 
The boomerang uniformity of $F$ is defined as $\cB_F:= \max \left\{\cB_F(a,b) \bigm|  a,b \in \mathbb{F}_{2^n}^*\right\}.$ However, it turns out that the BCT has limitations when evaluating boomerang switches in multiple rounds. This was observed by Wang and Peyrin~\cite{WP}, who showed that the BCT is not applicable effectively in some cases involving multiple rounds by giving an incompatibility example on two rounds of AES. To address this challenge, they introduced a new tool known as the Boomerang Difference Table (BDT) along with its variant BDT$^{\prime}$. Later, in 2020, Delaune et al.~\cite{DDV} renamed these tables as the Upper Boomerang Connectivity Table (UBCT) and the Lower Boomerang Connectivity Table (LBCT), respectively, to highlight the fact that UBCT and LBCT emphasize the upper and lower characteristic, respectively. The Feistel Boomerang Extended Table (FBET), or Boomerang Extended Table (BET) for SPN ciphers, proposed by Boukerrou et al.~\cite{Bouk}, has been renamed as the Extended Boomerang Connectivity Table (EBCT) by Delaune et al. in~\cite{DDV}. The entries in this table count the number of values such that the boomerang will return on a single S-box, with
all the differences fixed.

Further, to study the behavior of two consecutive S-boxes in the boomerang attack, Hadipour et al.~\cite{Hadi} introduced the notion of Double Boomerang Connectivity Table (DBCT). It may be noted that Yang et al.~\cite{YSS} illustrated the effect of an S-box on the probability of the 7-round boomerang distinguisher by comparing the DDT, BCT and DBCT entries of different S-boxes. Their findings highlight the importance of investigating the DBCT, in addition to the DDT and BCT, for evaluating an S-box's resistance to a boomerang attack. However, determining the DBCT entries for an S-box (or equivalently, for a vectorial Boolean function) is quite challenging, and rather limited research has been done in this direction.

Recently, Eddahmani and Mesnager~\cite{EM} studied various properties of the EBCT, LBCT and UBCT. They also determined these entries for the inverse function. In a separate work, Man et al.~\cite{MLLZ} also gave the DBCT, LBCT and UBCT entries for the inverse function. In this work, we discuss the DBCT entries of the Gold function, and we provide a connection between these new  EBCT, LBCT and UBCT concepts and the DDT entries for a differentially $\delta$-uniform function. As a consequence, it turns out that our result covers a particular case of a result of Eddahmani et al.~\cite{EM} and Man et al.~\cite{MLLZ} for the inverse function over $\F_{2^n}$. Moreover, we explicitly compute EBCT, LBCT and UBCT entries of three classes of differentially-4 uniform power permutations. 

We shall now give the structure of the paper. We first recall some definitions and results in Section~\ref{S2}. In Section~\ref{S3}, we provide some general results on the EBCT. We also give a generalization for the definitions of  the EBCT, LBCT and UBCT for any function over $\F_{2^n}$. In Section~\ref{oldS8}, we discuss the invariance of the EBCT, LBCT and UBCT under the CCZ-equivalence. In Section~\ref{S4}, we find a connection between the EBCT, LBCT and UBCT entries of a differentially $\delta-$uniform function and also express them in terms of the DDT entries. Section~\ref{S5} and Section~\ref{S6} deal with the EBCT, LBCT and UBCT entries of APN functions (leading to a characterization of APN functions) and 4-differential uniform functions, as well as some other well-known functions. We also computed the DBCT entries for the Gold function in Section~\ref{oldS7}. Finally, we conclude the paper in Section~\ref{S9}. 
 
\section{Preliminaries}\label{S2}

We will first recall several definitions and lemmas used in the subsequent sections. Throughout the paper, we shall use $\Tr_{m}^{n}$ to denote the (relative) trace function from $\F_{2^n} \rightarrow \F_{2^m}$, i.e., $\Tr_{m}^{n}(X) = \sum_{i=0}^{\frac{n-m}{m}} X^{2^{mi}}$, where $m$ and $n$ are positive integers and $m \mid n$. For $m = 1$, we use $\Tr$ to denote the absolute trace. As customary, for a set $S$, we let $a+S=\{a+s\,|\, s\in S\}$.

\begin{defn}\label{defE}\cite{DDV}
 Let $F$ be a permutation of $\F_{2^n}$. The Extended Boomerang Connectivity Table (EBCT) of $F$ is   a $2^n \times 2^n \times 2^n\times 2^n$ table where the entry at  $a, b, c, d \in \F_{2^n}$ is  
  \begin{equation*}\EB_F(a, b, c, d) = \left\lvert \left\{ X \in \F_{2^n}  \bigg| 
  \begin{cases} 
F(X)+F(X+a)=b\\
  F(X)+F(X+c)=d \\
  F^{-1}(F(X)+d)+F^{-1}(F(X+a)+d)=a
  \end{cases}
  \right\}\right \rvert .
\end{equation*} 
\end{defn}

\begin{defn}\label{defL}\cite{WP}
  Let $F$ be a permutation of $\F_{2^n}$. The Lower Boomerang Connectivity Table (LBCT) of $F$ is defined as a $2^n \times 2^n \times 2^n$ table where the entry at $a, b, c \in \F_{2^n}$ is  
  \begin{equation*}\LB_F(a, b, c) = \left\lvert \left\{ X \in \F_{2^n}  \bigg| 
  \begin{cases} 
  F(X)+F(X+b)=c \\
  F^{-1}(F(X)+c)+F^{-1}(F(X+a)+c)=a
  \end{cases}
  \right\}\right \rvert .
\end{equation*} 
\end{defn}

\begin{defn}\cite{WP}
  Let $F$ be a permutation of $\F_{2^n}$. The Upper Boomerang Connectivity Table (UBCT) of $F$ is defined as a $2^n \times 2^n \times 2^n$ table where the entry at  $a, b, c \in \F_{2^n}$ is  
  \begin{equation*}\UB_F(a, b, c) = \left\lvert \left\{ X \in \F_{2^n}  \bigg| 
  \begin{cases} 
  F(X)+F(X+a)=b \\
  F^{-1}(F(X)+c)+F^{-1}(F(X+a)+c)=a
  \end{cases}
  \right\}\right \rvert .
\end{equation*} 
  
\end{defn}
Recently, Eddahmani and Mesnager~\cite{EM} and, Man et al.~\cite{MLLZ} redefined the notions of LBCT and UBCT for a permutation $F$ without involving its compositional inverse (\cite{MLLZ} also mentions the non-necessity of the permutation property for the UBCT), as stated in the following lemmas. 

\begin{lem}\cite{EM}\label{L00}
 Let $F$ be a permutation of $\F_{2^n}$. Then for $a, b, c \in \F_{2^n}$, we have
 \begin{equation*}
 \LB_F(a, b, c) = \left\lvert \left\{ X \in \F_{2^n}   \bigg| \exists~ Y \in \F_{2^n} \mbox{ with }
  \begin{cases} 
  X+Y=b \\
  F(X+a)+F(Y+a)=c \\
  F(X)+F(Y)=c
  \end{cases}
  \right\}\right \rvert ,
\end{equation*} 
and 
 \begin{equation*}
 \UB_F(a, b, c) = \left\lvert \left\{ X\in \F_{2^n}   \bigg| \exists~ Y \in \F_{2^n} \mbox{ with }
  \begin{cases} 
  F(X+a)+F(Y+a)=c \\
  F(X)+F(Y)=c \\
  F(X)+F(X+a)=b
  \end{cases}
  \right\}\right \rvert .
\end{equation*} 
\end{lem}

\begin{defn}{(Feistel Boomerang Connectivity Table)}~\textup{\cite{Bouk}}
 Let $F:\F_{2^n}\to \F_{2^n}$ and $a,b \in \F_{2^n}$. The {\em Feistel Boomerang Connectivity Table} (FBCT) of $F$ is given by a $2^n \times 2^n$ table, of entry at $(a, b)\in \F_{2^n}\times\F_{2^n}$   defined by
 \begin{equation*}
\FB_F(a, b) =\left| \left\{X \in \F_{2^n}\bigm|  F(X + a + b) +F (X + b) +F(X + a) + F (X) = 0\right\}\right|.
\end{equation*}
This is also denoted as $\nabla_F(a,b)$. The maximum of the entries of the FBCT of $F$, for $a\neq b,ab\neq 0$, is also known as the  {\em second-order zero differential uniformity} of $F$, denoted as $\nabla_F$, which is defined over any finite field characteristic.
\end{defn}

\begin{rmk}  \label{r1} If $F$ is a permutation of $\F_{2^n}$, then it is immediate   that $$\FB_F(a, b) =\sum_{c\in\F_{2^n}}\LB_F(a, b, c).$$ 
Note that this implies that $\LB_F(a, b, c)\leq \FB_F(a, b),\ \forall a,b,c\in\F_{2^n}$.
\end{rmk}

\begin{defn}{(Double Difference Distribution Table)}~\textup{\cite{EM}}
 Let $F$ be a permutation of $\F_{2^n}$. The {\em double difference distribution
table $(\DD)$} of  $F$ is a $2^n\times2^n\times2^n$
 table where the entry at $(a, b, c)\in\F_{2^n}\times \F_{2^n}\times \F_{2^n}$   is given by
 \begin{equation*}
\DD_F (a, b, c) = \left| \left\{X \in \F_{2^n}\bigm|  F(X + a + b) +F (X + b) +F(X + a) + F (X) = c\right\}\right|.
\end{equation*}
Note that $\FB_F(a, b)=\DD_F(a,b,0)$.
\end{defn}

Moreover, in~\cite{EM} it was shown that for $a,b,c,d \in \F_{2^n}$ with $abcd=0$, the values of $\EB_F(a,b,c,d)$ are simple to compute if $F$ is a permutation, as are the values of $\LB_F(a, b, c)$ and $\UB_F(a, b, c)$, for $abc=0$. We recall these results in the following three lemmas.

\begin{lem}~\cite{EM}\label{L0001}
 Let $F$ be a permutation of $\F_{2^n}$. Then for $a,b,c,d \in \F_{2^n}$ with $abcd=0$, we have
 \allowdisplaybreaks
 \begin{align*}
 \EB_F(a, b, c, d) &=   
  \begin{cases} 
  2^n & \mbox{if}~ a=b=c=d=0, \\
\DDT_F(c,d)& \mbox{if}~ a=b=0,c\neq0,d\neq0,\\
\DDT_F(a,b)& \mbox{if}~ a\neq0,b\neq0,c=d=0,\\
  0 & \mbox{otherwise,}
  \end{cases}\\
 \LB_F(a, b, c) &=  
  \begin{cases} 
  2^n & \mbox{if}~ b=c=0, \\
  0 & \mbox{if}~ b \neq 0, c = 0, \\
 \DDT_F(b,c) & \mbox{if}~ a=0, \\
  0 & \mbox{if}~ b=0, c \neq 0, 
  \end{cases} \\
  \UB_F(a, b, c) &= 
  \begin{cases} 
  2^n & \mbox{if}~ a=b=0, \\
  0 & \mbox{if}~ a=0, b \neq 0, c \neq 0, \\
  \DDT_F(a,b) & \mbox{if}~ c=0, \\
  0 & \mbox{if}~ a \neq 0, b=0, c \neq 0.
  \end{cases} 
\end{align*} 
\end{lem}

\begin{rmk} If $F$ is not a permutation, some of the results of Lemma~\textup{\ref{L0001}} are not true. We give the generalization of  Lemma~\textup{\ref{L0001}} in Section~\textup{\ref{S4}}.
\end{rmk}

Also,  Wagner~\cite{Wagner} shows that $\UB_F(a,b,b)=\DDT_F(a,b)$, as well as,  $ \displaystyle \text{BCT}_F(a,c)=\sum_{b\in\F_{2^n}}\UB_F(a,b,c)$. Additionally, it is easy to see that $\DDT_F(a,c)=\LB_F(a,a,c)$. Next, we recall the notion of DBCT as follows.

 \begin{defn} \cite{YSS}
Let $F(X)$ be a mapping from $\F_{2^n}$ to itself. The {\em Double Boomerang Connectivity Table} (DBCT) is a $2^n \times 2^n$ table defined for $(a, d) \in \F_{2^n}^{2}$ by
\begin{equation*}
 \DB_F (a, d) = \sum_{b,c}\text{dbct}(a,b,c,d),
\end{equation*}
where dbct$(a, b, c, d) = \UB_F(a, b, c)  \LB_F (b, c, d)$. For $a = 0$ or $d = 0$, it can be easily obtained that
\allowdisplaybreaks
\begin{align*}
\DB_{F}(0, d) &= \sum_{c} \UB_{F} (0, 0, c)  \LB_{F} (0, c, d) = 2^{2n},\text{ and}\\
\DB_{F}(a, 0) &=  \sum_{b}
\UB_{F} (a, b, 0) \LB_{F} (b, 0, 0) = 2^{2n}.
\end{align*}
\end{defn}

\section{General results on the EBCT}\label{S3}

In a similar way as the results in \cite{EM} and ~\cite{MLLZ}, we can redefine the EBCT for a permutation $F$ without involving its compositional inverse, as stated in the following lemma.

\begin{lem}\label{E1}
 Let $F$ be a permutation of $\F_{2^n}$. Then for $a, b, c,d \in \F_{2^n}$, we have
 \begin{equation*}\EB_F(a, b, c,d) = \left\lvert \left\{ X \in \F_{2^n}  \bigg| 
  \begin{cases} 
 F(X)+F(X+a)=b\\
F(X)+F(X+c)=d\\
F(X+a+c)+F(X+a)=d
  \end{cases}
  \right\}\right \rvert .
\end{equation*} 
\end{lem}
\begin{proof}
We write
  \begin{equation*}
  \EB_F(a, b, c, d) = \left\lvert \left\{ X \in \F_{2^n}  \bigg| 
  \begin{cases} 
(1)~ F(X)+F(X+a)=b\\
(2)~ F(X)+F(X+c)=d \\
(3)~ F^{-1}(F(X)+d)+F^{-1}(F(X+a)+d)=a
  \end{cases}
  \right\}\right \rvert .
\end{equation*} 
Equation (2) implies $F(X)+d=F(X+c)$. Equation (3) is then equivalent to
\begin{align*}
&F^{-1}(F(X+c))+F^{-1}(F(X+a)+d)=a,\\
&F^{-1}(F(X+a)+d)=X+a+c,\\
&F(X+a)+d=F(X+a+c),
\end{align*}
completing the proof of the lemma.
\end{proof}

As a consequence, we can derive the connection between the EBCT of a permutation and the EBCT of its compositional inverse.

\begin{cor}
Let $F$ be a permutation of $\F_{2^n}$. Then for $a, b, c,d \in \F_{2^n}$, we have \\
$\EB_F(a, b, c,d)=\EB_{F^{-1}}(b,a,d,c)$.
\end{cor}
\begin{proof} We have by Lemma \ref{E1} that $\EB_F(a, b, c,d)$ is given by the following equations:
 \begin{equation*}\EB_F(a, b, c,d) = \left\lvert \left\{ X \in \F_{2^n}  \bigg| 
  \begin{cases} 
(1)~ F(X)+F(X+a)=b\\
(2)~ F(X)+F(X+c)=d\\
(3)~ F(X+a+c)+F(X+a)=d
  \end{cases}
  \right\}\right \rvert .
\end{equation*} 
Let $Y=F(X)$. Then, Equation (1) gives $Y+F(X+a)=b$, which is equivalent to $F(X+a)=Y+b$, and then to $X+a=F^{-1}(Y+b)$, rendering the equation $F^{-1}(Y+b)+F^{-1}(Y)=a$.
Similarly, Equation (2) is equivalent to $F^{-1}(Y+d)+F^{-1}(Y)=c$.

Lastly, inserting Equation (1) in Equation (3), we obtain $F(X+a+c)=Y+d+b$, which is equivalent to $F^{-1}(Y+b+d)=X+a+c$. Since  $F^{-1}(Y+b)+F^{-1}(Y)=F^{-1}(Y+b)+X=a$, this gives in  turn $F^{-1}(Y+b+d)=F^{-1}(Y+b)+c$, equivalent to $F^{-1}(Y+b+d)+F^{-1}(Y+b)=c$. 
We therefore get the claim.
\end{proof}

Lemmas~\ref{L00} and \ref{E1} also suggest a generalization of these concepts for the set of all functions over $\F_{2^n}$ (not necessarily permutations). 
\begin{defn} 
\label{generaldef} 
Let $F$ be a function over $\F_{2^n}$. Then, for $a, b, c,d \in \F_{2^n}$,  we can define the EBCT, LBCT and UBCT in the following way,
\allowdisplaybreaks
 \begin{align}
 \EB_F(a, b, c,d) &= \left\lvert \left\{ X \in \F_{2^n} \bigg|
  \begin{cases}\label{EBCT}
 F(X)+F(X+a)=b\\
F(X)+F(X+c)=d\\
F(X+a+c)+F(X+a)=d
  \end{cases}
  \right\}\right \rvert ,\\
\LB_F(a, b, c) &= \left\lvert \left\{ X \in \F_{2^n}  \bigg| \exists Y \in \F_{2^n} \mbox{ with } 
  \begin{cases} \label{LBCT}
  X+Y=b\\
  F(X+a)+F(Y+a)=c \\
  F(X)+F(Y)=c
  \end{cases}
  \right\}\right \rvert,\\
  \UB_F(a, b, c) &= \left\lvert \left\{ X \in \F_{2^n}\bigg| \exists Y \in \F_{2^n} \mbox{ with } 
  \begin{cases} \label{UBCT}
  F(X+a)+F(Y+a)=c \\
  F(X)+F(Y)=c \\
  F(X)+F(X+a)=b
  \end{cases}
  \right\}\right \rvert .
\end{align} 
\end{defn}

\begin{rmk} 
Note that if $F$ is a permutation, these definitions are equivalent to the original definitions of these concepts. We preferred to define these concepts this way to avoid overcounting, but we point out that even if we count the number of pairs $(X,Y)$ in each set, all results hold, except for the UBCT in Theorem~\textup{\ref{deltadiff}}, where its expression becomes too cumbersome.
\end{rmk}

\section{Invariance under the CCZ, extended affine and affine-equivalence}
\label{oldS8}

We recall that two functions $F,G:\F_{2^n}\to \F_{2^n}$ are {\em CCZ-equivalent} \cite{CCZ98} if there exists an affine permutation  $\cA$ on $\F_{2^n}\times \F_{2^n}$ such that 
\[
\displaystyle \left\{\colvec{x}{G(x)}\,|\,x\in\F_{2^n} \right\}=\left\{\cA  \colvec{x}{F(x)} \,|\,x\in\F_{2^n}\right\}.
\]
 As customary, we use the natural identification of  the elements in $\F_{2^n}$ with the elements in~$\F_2^n$, and, by abuse, we denote by $x$ both an element in $\F_{2^n}$ and the corresponding element in~$\F_2^n$. We also decompose the affine permutation $\cA$ as an affine block-matrix,  for an input vector $u\in\F_{2^n}\times \F_{2^n}$: 
$\cA{u}=\begin{pmatrix} 
\cA_{11} &\cA_{12}\\
\cA_{21} &\cA_{22}
\end{pmatrix} {u}+\colvec{C}{D}$, where $\cA_{ij}:\F_{2^n}\to\F_{2^n}, i,j \in \{1,2\}$ and~$\colvec{C}{D}$ is a column vector in $\F_{2^{n}}\times \F_{2^n}$.

When $\cA_{12}=0$, we say that $F$ and $G$ are {\em extended affine (EA)}-equivalent. This can be also written as $G=F_2\circ F\circ F_1+F_0$, for some $F_i(X)=L_i(X)+B_i$ affine functions, where $L_i$ are linearized polynomials, $B_i\in\F_{2^n}$ such that $F_1$ and $F_2$ are permutation polynomials.

When $\cA_{12}=\cA_{21}=0$, we say that $F$ and $G$ are {\em affine}-equivalent.  This can be also written as $G=F_2\circ F\circ F_1$, for some $F_i(X)=L_i(X)+B_i$ affine functions, where $L_i$ are linearized polynomials, $B_i\in\F_{2^n}$ such that $F_1$ and $F_2$ are permutation polynomials.

Here, we discuss the behavior of the EBCT, LBCT and UBCT of a function $F$ under CCZ, extended affine and affine-equivalence.  It is worth noting that Eddahmani and Mesnager~\cite{EM} proved that the EBCT, LBCT and UBCT of a permutation remain invariant under affine-equivalence.

 \begin{thm} 
 Given any functions $F,G$ on $\F_{2^n}$ that are CCZ-equivalent via the  affine permutation  $\cA$ on $\F_{2^{2n}}$, and given any $a,b,c,d\in\F_{2^n}$, then $\EB_F(a,b,c,d)=\EB_G(\alpha,\beta,\gamma,\delta)$, with $\alpha=\cA_{12}b+\cA_{11}a,\,\beta=\cA_{22}b+\cA_{21}a,\,\gamma=\cA_{12}d+\cA_{11}c\mbox{ and }\delta=\cA_{22}d+\cA_{21}c$,
and thus, the EBCT spectrum is preserved under the CCZ-equivalence. 
\end{thm}

\begin{proof}
Let $\displaystyle \left\{\colvec{x}{G(x)},x\in\F_{2^n} \right\}=\left\{\cA  \colvec{x}{F(x)} ,x\in\F_{2^n} \right\}.$ Then, for every $X,c\in\F_{2^n}$ there exists $y,z\in\F_{2^n}$ such that  $\colvec{y}{G(y)}=\cA  \colvec{X}{F(X)}$ and $\colvec{z}{G(z)}=\cA  \colvec{X+c}{F(X+c)}$, which gives 
 \allowdisplaybreaks
\begin{align*}
y&=\cA_{11}X+\cA_{12}F(X)+C,\\
G(y)&=\cA_{21}X+\cA_{22}F(X)+D,\\
z&=\cA_{11}(X+c)+\cA_{12}F(X+c)+C,\\
G(z)&=\cA_{21}(X+c)+\cA_{22}F(X+c)+D.
\end{align*}

Suppose $F(X+c)+F(X)=d$. Then, 
 \allowdisplaybreaks
\begin{align*}
G(z)+G(y)&=\cA_{21}(X+c)+\cA_{22}F(X+c)+D+\cA_{21}X+\cA_{22}F(X)+D\\
&=\cA_{22}(F(X+c)+F(X))+\cA_{21}c=\cA_{22}d+\cA_{21}c.
\end{align*}

 Taking $z=y+\gamma$ and $\delta=\cA_{22}d+\cA_{21}c$, we obtain:
$G(y+\gamma)+G(y)=\delta.$
We can also see that 
 \allowdisplaybreaks
\begin{align*}
\gamma=y+z&=\cA_{11}X+\cA_{12}F(X)+C+\cA_{11}(X+c)+\cA_{12}F(X+c)+C\\
&=\cA_{12}(F(X+c)+F(X))+\cA_{11}c=\cA_{12}d+\cA_{11}c.
\end{align*}
Therefore,
$
 F(X)+F(X+c)=d\Rightarrow G(y+\gamma)+G(y)=\delta, 
$
where $\gamma=\cA_{12}d+\cA_{11}c\mbox{ and }\delta=\cA_{22}d+\cA_{21}c.$

Similarly, taking $\colvec{u}{G(u)}=\cA  \colvec{X+a}{F(X+a)}$, we can prove that 
$
F(X)+F(X+a)=b\Rightarrow G(y+\alpha)+G(y)=\beta, $
where $\alpha=\cA_{12}b+\cA_{11}a\mbox{ and }\beta=\cA_{22}b+\cA_{21}a$.

Now, we need to see if the third equation, $F(X+a)+F(X+a+c)=d$, transforms into $G(y+\alpha)+G(y+\alpha+\gamma)=\delta$ with the same values of $\alpha, \gamma,$ and $\delta$.
Taking $\colvec{w}{G(w)}=\cA  \colvec{X+a+c}{F(X+a+c)}$,   we see that 
\begin{align*}
w&=\cA_{11}(X+a+c)+\cA_{12}F(X+a+c)+C,\\
G(w)&=\cA_{21}(X+a+c)+\cA_{22}F(X+a+c)+D,\\
u&=\cA_{11}(X+a)+\cA_{12}F(X+a)+C,\\
G(u)&=\cA_{21}(X+a)+\cA_{22}F(X+a)+D.
\end{align*}
Then,
\begin{align*}
G(w)+G(u)&=\cA_{21}(X+a+c)+\cA_{22}F(X+a+c)+D+\cA_{21}(X+a)+\cA_{22}F(X+a)+D\\
&=\cA_{22}(F(X+a+c)+F(X+a))+\cA_{21}c=\cA_{22}d+\cA_{21}c=\delta.
\end{align*}

As before, $u=y+\alpha$, with 
 \allowdisplaybreaks
 \begin{align*}
 \alpha=y+u&=\cA_{11}X+\cA_{12}F(X)+C+\cA_{11}(X+a)+\cA_{12}F(X+a)+C\\
 &=\cA_{12}b+\cA_{11}a,\text{ and}\\
  w+u&=\cA_{11}(X+a+c)+\cA_{12}F(X+a+c)+C+\cA_{11}(X+a)+\cA_{12}F(X+a)+C\\
  &=\cA_{12}(F(X+a+c)+F(X+a))+\cA_{11}c=\cA_{12}d+\cA_{11}c=\gamma.
  \end{align*} 
  Therefore, we get the desired equation
$G(y+\alpha)+G(y+\alpha+\gamma)=\delta,$
with the same values of $\alpha, \gamma,$ and $\delta$.

Finally, we note that $$\left(\begin{array}{c}
\alpha\\
\beta\\
\gamma\\
\delta\end{array}\right)=\left(\begin{array}{cccc}\cA_{11}&\cA_{12}&0&0\\
\cA_{21}&\cA_{22}&0&0\\
0&0&\cA_{11}&\cA_{12}\\
0&0&\cA_{21}&\cA_{22}\end{array}\right)\left(\begin{array}{c}
a\\
b\\
c\\
d\end{array}\right)=\cB\left(\begin{array}{c}
a\\
b\\
c\\
d\end{array}\right)$$
and, since $\cA$ is invertible, the matrix $\cB$   is also invertible.
This concludes the proof.
\end{proof}

Below we use again this section's notations.
\begin{thm} 
The LBCT spectrum is preserved under EA-equivalence. If $F$ and $G$ are EA-equivalent, then $\LB_F(a,b,c)=\LB_G(\alpha_L,\beta_L,\gamma_L)$, where $\alpha_L=\cA_{11}a$,$\beta_L=\cA_{11}b\mbox{ and }\gamma_L=\cA_{22}c+\cA_{21}b$. \end{thm}


\begin{proof} 
Here, since $  X+Y=b$, we can take $Y=X+b$, and write the second and third equations of the system of $\LB_F(a,b,c)$, respectively, as $F(X+a)+F(X+b+a)=c$ and $F(X)+F(X+b)=c$.

There exist $v,t\in\F_{2^n}$ such that $\colvec{v}{G(v)}=\cA  \colvec{X+b}{F(X+b)}$ and $\colvec{t}{G(t)}=\cA  \colvec{X+b+a}{F(X+b+a)}$. 
Similarly as in the proof of the previous theorem, 
\begin{align*}
v&=\cA_{11}(X+b)+\cA_{12}F(X+b)+C,\\
G(v)&=\cA_{21}(X+b)+\cA_{22}F(X+b)+D, 
\end{align*}
gives that 
\begin{equation*}
F(X)+F(X+b)=c\Rightarrow G(y+\beta_L)+G(y)=\gamma_L, 
\end{equation*}
where $\beta_L=\cA_{12}c+\cA_{11}b\mbox{ and }\gamma_L=\cA_{22}c+\cA_{21}b$.

From
\begin{align*}
t&=\cA_{11}(X+a+b)+\cA_{12}F(X+a+b)+C,\\
G(t)&=\cA_{21}(X+a+b)+\cA_{22}F(X+a+b)+D,
\end{align*}
we obtain 
$
F(X+a)+F(X+a+b)=c\Rightarrow G(y+\alpha_L)+G(y+\alpha_L+\beta_L)=\gamma_L, $
where $\alpha_L=\cA_{12}(F(X)+F(X+a))+\cA_{11}a\mbox{ and }\gamma_L=\cA_{22}c+\cA_{21}b,\,\beta_L=\cA_{12}c+\cA_{11}b$.

We see that, in general, $\alpha_L$ is not a constant. However, if $\cA_{12}=0$ (i.e. EA-equivalence), we have that $\alpha_L$ is a constant. Therefore, if $F$ and $G$ are EA-equivalent, then $\LB_F(a,b,c)=\LB_G(\alpha_L,\beta_L,\gamma_L)$, where $\alpha_L=\cA_{11}a$,$\beta_L=\cA_{11}b\mbox{ and }\gamma_L=\cA_{22}c+\cA_{21}b$.

Finally,  $$\left(\begin{array}{c}
\alpha_L\\
\beta_L\\
\gamma_L\end{array}\right)=\left(\begin{array}{ccc}\cA_{11}&0&0\\
0&\cA_{11}&0\\
0&\cA_{21}&\cA_{22}\end{array}\right)\left(\begin{array}{c}
a\\
b\\
c\end{array}\right)=\cC\left(\begin{array}{c}
a\\
b\\
c\end{array}\right),$$
and, since $\cA$ is invertible, the matrix $\cC$   is also invertible. 
This concludes the proof.
\end{proof}

We use the same notations as in the previous two theorems.
\begin{thm} 
The UBCT spectrum is preserved under affine equivalence. If $F$ and $G$ are affine equivalent, then $\UB_F(a,b,c)=\UB_G(\alpha_U,\beta_U,\gamma_U)$, where $\alpha_U=\cA_{11}a,\,\beta_U=\cA_{22}b\mbox{ and }\gamma_U=\cA_{22}c$.
\end{thm}
\begin{proof} 
The relevant equations are here $F(X+a)+F(Y+a)=c$, $ F(X)+F(Y)=c$ and $F(X)+F(X+a)=b$.
As before,
\begin{equation*}
F(X)+F(X+a)=b\Rightarrow G(y+\alpha_U)+G(y)=\beta_U,  
\end{equation*}
where $\alpha_U=\cA_{12}b+\cA_{11}a\mbox{ and }\beta_U=\cA_{22}b+\cA_{21}a$.

Taking as before $\colvec{y}{G(y)}=\cA  \colvec{X}{F(X)}$ and $\colvec{u}{G(u)}=\cA  \colvec{X+a}{F(X+a)}$, and taking now $\colvec{y'}{G(y' )}=\cA  \colvec{Y}{F(Y)}$, $\colvec{u'}{G(u')}=\cA  \colvec{Y+a}{F(Y+a)}$, we obtain that 
\allowdisplaybreaks
\begin{align*}
y&=\cA_{11}X+\cA_{12}F(X)+C,\\
G(y)&=\cA_{21}X+\cA_{22}F(X)+D,\\
u&=\cA_{11}(X+a)+\cA_{12}F(X+a)+C,\\
G(u)&=\cA_{21}(X+a)+\cA_{22}F(X+a)+D,\\
y'&=\cA_{11}Y+\cA_{12}F(Y)+C,\\
G(y')&=\cA_{21}Y+\cA_{22}F(Y)+D,\\
u'&=\cA_{11}(Y+a)+\cA_{12}F(Y+a)+C,\\
G(u')&=\cA_{21}(Y+a)+\cA_{22}F(Y+a)+D.
\end{align*}
Then, 
\begin{equation*}
G(y)+G(y')=\cA_{21}(X+Y)+\cA_{22}(F(X)+F(Y))=\cA_{21}(X+Y)+\cA_{22}c=\gamma_U.
\end{equation*}
Here we see that, in general, $\gamma_U$ is not constant.
Similarly,
\begin{align*}
G(u)+G(u')&=\cA_{21}(X+Y)+\cA_{22}(F(X+a)+F(Y+a))=\cA_{21}(X+Y)+\cA_{22}c=\gamma_U.
\end{align*} 
However, if $\cA_{21}=0$, $\gamma_U$ is constant. Under affine equivalence,  $\cA_{12}=\cA_{21}=0$, and $\UB_F(a,b,c)=\UB_G(\alpha_U,\beta_U,\gamma_U)$, with $\gamma_U=\cA_{22}c$, $\beta_U=\cA_{22}b\mbox{ and }\alpha_U=\cA_{11}a$.

Finally,  $$\left(\begin{array}{c}
\alpha_U\\
\beta_U\\
\gamma_U\end{array}\right)=\left(\begin{array}{ccc}\cA_{11}&0&0\\
0&\cA_{22}&0\\
0&0&\cA_{22}\end{array}\right)\left(\begin{array}{c}
a\\
b\\
c\end{array}\right)=\cD\left(\begin{array}{c}
a\\
b\\
c\end{array}\right),$$
and, since $\cA$ is invertible, the matrix $\cD$   is also invertible. 
This concludes the proof.
\end{proof}
The following example shows that the UBCT spectrum is in general not preserved under CCZ equivalence.
\begin{exmp}
Define $F(X)=X^9$ and $G(X)=X^9+(X^8+X)\Tr(X^9+X)$ over $\F_{2^5}$. Budhagyan et al.~\cite{BCP} showed that while $F(X)$ and $G(X)$ are CCZ equivalent, they are not EA equivalent. Despite of them being CCZ equivalent, their UBCT spectrum differs. Notably, there are 
$992$ tuples $(a,b,c) \in \F_{2^n} \times \F_{2^n} \times \F_{2^n}$ such that $\UB_F(a,b,c)=2$, while the number of tuples $(a,b,c) \in \F_{2^n} \times \F_{2^n} \times \F_{2^n}$ for which $\UB_G(a,b,c)=2$ is $982$. 
\end{exmp}
We now present an example illustrating that the UBCT spectrum is in general not even preserved under EA equivalence.
\begin{exmp}
Let $F(X)=X^5$ and $G(X)=X^5+X$ be defined over $\F_{2^3}$. Notice that $F(X)$ is a permutation over $\F_{2^3}$ while $G(X)$ is a non-permutation over $\F_{2^3}$. Moreover, $F(X)$ and $G(X)$ are EA equivalent. However, the UBCT spectrum is not preserved under EA equivalence.  Specifically, the number of tuples $(a,b,c) \in \F_{2^n} \times \F_{2^n} \times \F_{2^n}$ for which $\UB_F(a,b,c)=2$ is $448$, whereas the number of tuples $(a,b,c) \in \F_{2^n} \times \F_{2^n} \times \F_{2^n}$ satisfying $\UB_G(a,b,c)=2$ is $452$. 
\end{exmp}

\section{EBCT, LBCT and UBCT of a function in terms of its DDT}
\label{S4}

In this section, we present a general result that gives an intriguing connection between the EBCT, LBCT and UBCT entries of a differentially  $\delta$-uniform function $F$ (not necessarily a permutation) with its DDT entries, and between the entries of the EBCT and the LBCT and the UBCT. It also includes a generalization of Lemma~\ref{L0001} to arbitrary functions (not necessarily permutations). 

\begin{thm}
\label{deltadiff} 
Let $F$ be a differentially $\delta$-uniform function on $\F_{2^n}$, and let $a,b,c,d\in\F_{2^n}$. Let  $k \in \left\{1,2,\ldots,\frac{\delta}{2}\right\}$. If $\DDT_F(a,b)=2k$, we let $\{x_1,x_1+a,x_2,x_2+a,\ldots,x_{k},x_{k}+a\}$  denote the distinct solutions of the equation $F(X+a)+F(X)=b$; if $\DDT_F(b,c)=2k$, we let $\{y_1,y_1+b,y_2,y_2+b,\ldots,y_{k},y_{k}+b\}$ denote the distinct solutions of  the equation $F(X+b)+F(X)=c$; if $\DDT_F(c,d)=2k$, let $\{z_1,z_1+c,z_2,z_2+c, \ldots,z_{k},z_{k}+c\}$ be the distinct solutions of the equation $F(X+c)+F(X)=d$. Then we have,
\begin{equation*}\EB_{F}(a, b, c,d) = 
  \begin{cases}
\DDT_F(a,b) & \mbox{if } c = d =0, \\ 
  \DDT_{F}(c,d) & \mbox{if}~ ac \neq 0, a=c \mbox{ and } b=d ; \mbox{ or } a=0, b=0 \mbox{ and } c \neq 0,\\
4\ell &\mbox{if } a\neq c, ac \neq 0 \mbox{ and } \DDT_{F}(c,d)=2k=4r \mbox{~or~} 4r+2,  
\\
&~\mbox{where}~r>0\mbox{~is an integer,}\\
&~1 \leq \ell \leq r~\mbox{ is the largest integer such that} \\
&~(a,b) \in U(i_1, j_1) \cap U(i_2,j_2) \cap \cdots \cap U(i_{\ell}, j_{\ell}),\\
& ~ 1\leq i_1,i_2,\ldots, i_{\ell},j_1,j_2,\ldots, j_{\ell}\leq k~\mbox{are distinct integers,}\\
0 & \mbox{otherwise},
  \end{cases}
\end{equation*} 
where $U(i,j)= \{\left(z_i+z_j, F(z_i)+F(z_j)\right), \left(z_i+z_j+c, F(z_i)+F(z_j)+d\right) ~\mid ~ 1\leq i \neq j\leq k\} \subseteq \F_{2^n}^{*} \times \F_{2^n}^{*}.$
Moreover, we have
\begin{equation*}\LB_{F}(a, b, c) = 
  \begin{cases}  
  \DDT_{F}(b,c) & \mbox{if}~b=c=0; \mbox{ or }a=b \mbox{ and } ab\neq0; \mbox{ or } a=0 \mbox{ and } b \neq 0,\\ 
4\ell &\mbox{if } a\neq b,  ab\neq0 \mbox{ and } \DDT_{F}(b,c)=2k=4r \mbox{~or~} 4r+2, \\ &~\mbox{where}~r>0\mbox{~is an integer,}\\
&~1 \leq \ell \leq r~\mbox{ is the largest integer such that} \\
&~a \in V(i_1, j_1) \cap V(i_2,j_2) \cap \cdots \cap V(i_{\ell}, j_{\ell}),\\
& ~ 1\leq i_1,i_2,\ldots, i_{\ell},j_1,j_2,\ldots, j_{\ell}\leq k~\mbox{are distinct integers,}\\
  0 & \mbox{otherwise,}
  \end{cases}
\end{equation*} 
where $V(i,j)= \{y_i+y_j, y_i+y_j+b ~\mid ~ 1\leq i \neq j\leq k\} \subseteq \F_{2^n}^{*}$. Finally, 
\begin{equation*}\UB_{F}(a, b, c) = 
  \begin{cases} 
|F^{-1}(c+\Im(F))|& \mbox{if}~a= b= 0, \\
  \DDT_{F}(a,b) & \mbox{if}~ b=c \mbox{ and } a \neq 0, \\ 
4\ell &\mbox{if } c \neq b, a \neq 0 \mbox{ and } \DDT_{F}(a,b)=2k=4r \mbox{~or~} 4r+2, \\
&~\mbox{where}~r>0\mbox{~is an integer,}\\
&~1 \leq \ell \leq r~\mbox{ is the largest integer such that} \\
&~c \in W(i_1, j_1) \cap W(i_2,j_2) \cap \cdots \cap W(i_{\ell}, j_{\ell}),\\
& ~ 1\leq i_1,i_2,\ldots, i_{\ell},j_1,j_2,\ldots, j_{\ell}\leq k~\mbox{are distinct integers,}\\
  0 & \mbox{otherwise,}
  \end{cases}
\end{equation*} 
where $W(i,j)= \{F(x_i)+F(x_j), F(x_i)+F(x_j)+b ~\mid ~ 1\leq i \neq j\leq k\} \subseteq \F_{2^n}^{*},$  $\Im(F)$ is the image of $F$ and $F^{-1}(\cdot)$ denotes the preimage of the argument.
\end{thm}
\begin{proof}
It is easy to see that if $X$ is a solution to System~\eqref{EBCT}  at $(a,b,c)$  (here and throughout, for easy writing, we refer to the system in the definition of such a cardinality $(\star)$ as System~$(\star)$), then $\{X,X+c,X+a,X+a+c\}$ are solutions of $F(X)+F(X+c)=d$. Since $F$ is differentially $\delta$-uniform, the equation $F(X)+F(X+c)=d$ can have zero or $2k$ solutions, where $k\in \left\{1,2,\ldots,\frac{\delta}{2}\right\}$. If $\DDT_F(c,d)=0$, then System~\eqref{EBCT} has no solutions, and therefore $\EB_F(a,b,c,d)=0$. 

Furthermore, it is straightforward to observe that  if $a=0$ (or $c=0$), then System~\eqref{EBCT} has a solution only when $b=0$ (or $d=0$, respectively) and otherwise, no solution. We can now divide the proof into the following five cases as follows

\noindent\textbf{Case 1.}  If $a=c=0$, then $b=d=0$. Therefore for $a=b=c=d=0$, we get $\EB_F(a,b,c,d)=2^n=\DDT_F(0,0)$.

\noindent\textbf{Case 2.}  If $a \neq 0, c=d=0$, we obtain $\EB_F(a,b,c,d)=\DDT_F(a,b)$.

\noindent\textbf{Case 3.}  If $a=b=0, c \neq 0$, then $\EB_F(a,b,c,d)=\DDT_F(c,d)$.

\noindent\textbf{Case 4.} If $a=c$ and $ac \neq 0$ then from the first two equations in System~\eqref{EBCT}, we obtain $b=d$. Thus, $a=c$ and $b=d$ imply that $\EB_F(a,b,c,d)=\DDT_F(c,d)$. 

It is easy to observe that when $\DDT_F(c,d)=2$, System~\eqref{EBCT} will have a solution only if $a=c$ or $a=0$.

\noindent\textbf{Case 5.} Let $a \neq c, ac\neq0$ and $\DDT_F(c,d)=2k$, where solutions of $F(X)+F(X+c)=d$ are from the set $S(c,d):=\{z_1,z_2,\ldots,z_k,z_1+c,z_2+c,\ldots, z_k+c\}$. Due to the second equation in System~\eqref{EBCT}, all solutions of this system must necessarily come from $S(c,d)$ itself. It follows that if $z_i$ is solution to System~\eqref{EBCT}, then from the third equation, we have $z_i+a \in S(c,d) \setminus \{z_i, z_i+c\}$, or equivalently, $a \in z_i+S(c,d) \setminus \{z_i, z_i+c\}$. Furthermore, from the first equation, we have $b=F(z_i)+F(z_i+a)$. We consider the following set in $\F_{2^n}^{*} \times \F_{2^n}^{*}$, 
$$U(i,j)= \{\left(z_i+z_j, F(z_i)+F(z_j)\right), \left(z_i+z_j+c, F(z_i)+F(z_j)+d\right) ~\mid ~ 1\leq i \neq j\leq k\}.$$ 
It easy to check that $z_i,z_j, z_i+c,z_j+c$ are four distinct solutions of System~\eqref{EBCT} if and only if $(a, b)\in U(i,j)$. Note that for a given pair $(a,b) \in U (i,j)$, there may be additional solutions to System~\eqref{EBCT} beyond these four. Thus, we have $\EB_F(a,b,c,d) \geq 4$ if $(a,b) \in U (i,j)$ and $\DDT_F(c,d)=2k$. For $k \geq 4$ and $1 \leq i_1 \neq j_1 \neq i_2 \neq j_2 \leq k$, let $(a,b) \in U(i_1, j_1)$ and $(a', b') \in U(i_2, j_2)$. If $(a,b)=(a',b')$, then System~\eqref{EBCT} has at least eight solutions, namely, $z_{i_1}, z_{i_2}, z_{j_1}, z_{j_2}, z_{i_1}+c, z_{i_2}+c, z_{j_1}+c, z_{j_2}+c$; otherwise, it will have at least four solutions. Following a similar argument, it is clear that if the pairs $(a,b)$ are equal for three different $U(i,j)$, then System~\eqref{EBCT} will have at least twelve solutions. However, this process must eventually stop, depending on the value of $\DDT_F(c,d)$, which we will describe below. 

Depending on whether $2k \equiv 0 \pmod4$ or $2k \equiv 2 \pmod 4$, we can express $\DDT_F(c,d)$ as $\DDT_F(c,d)=4r$ or $\DDT_F(c,d)=4r+2$, respectively, where $r\geq 0$ is an integer. The case $k=1$ has already been addressed in Case 1. For $k=2$ or $k=3$, we have $\EB_F(a,b,c,d) \in \{0,4\}$. Now we consider the case $k \geq 4$. Let $1 \leq \ell \leq r$ be the largest integer for which there are distinct integers $i_1,i_2,\ldots, i_{\ell},j_1,j_2,\ldots, j_{\ell} \in \{1,2, \ldots, k \}$ such that $(a,b) \in U(i_1, j_1) \cap U(i_2,j_2) \cap \cdots \cap U(i_{\ell}, j_{\ell})$, then System~\eqref{EBCT} will have $4 \ell$ solutions. This proves the result for the EBCT.

Next, we compute the LBCT entries  at $(a,b,c)$  for a function $F$, using System~\eqref{LBCT}. If $(X,Y) \in \F_{2^n} \times \F_{2^n}$ satisfies System~\eqref{LBCT}, then $Y=X+b$ and $\{X,X+b,X+a,X+a+b\}$ are solutions of $F(X)+F(X+b)=c$. As $F$ is differentially $\delta$-uniform, the equation $F(X)+F(X+b)=c$ can have zero or $2k$ solutions, where $k\in \left\{1,2,\ldots,\frac{\delta}{2}\right\}$. Clearly, when $\DDT_F(b,c)=0$, System~\eqref{LBCT} has no solutions, and therefore $\LB_F(a,b,c)=0$. 

Notice that for $b=0$, System~\eqref{LBCT} has a solution only if $c=0$ and otherwise, no solution. Moreover, for $b=c=0$, System~\eqref{LBCT} has $2^n$ solutions ($=\DDT_F(b,c)$) regardless of whether $a=0$ or $a \neq 0$. We will now split our analysis into the following cases when $b \neq 0$:

\noindent\textbf{Case 1.} Let $a=0, b\neq 0$, then $\LB_F(a,b,c)=\DDT_F(b,c)$.

\noindent\textbf{Case 2.} Consider $a=b, ab \neq 0$, then $\LB_F(a,b,c)=\DDT_F(b,c)$. 

Further, one can see that if $\DDT_F(b,c)=2$, System~\eqref{LBCT} will have a solution only if $a=b$ or $a=0$. 

\noindent\textbf{Case 3.} Next, consider $a \neq b, ab \neq 0$ and $\DDT_F(b,c)=2k$, where solutions of $F(X)+F(X+b)=c$ are from the set $S(b,c):=\{y_1,y_2,\ldots,y_k,y_1+b,y_2+b,\ldots, y_k+b\}$. If $X=y_i$ is solution to System~\eqref{LBCT}, then first equation implies that $Y=y_i+b$ and the second equation implies that $y_i+a \in S(b,c) \setminus \{y_i, y_i+b\}$, or equivalently, $a \in y_i+S(b,c) \setminus \{y_i, y_i+b\}$. We now consider the following set
 $$
 V(i,j)= \{y_i+y_j, y_i+y_j+b ~\mid ~ 1\leq i \neq j\leq k\} \subseteq \F_{2^n}^{*}. 
 $$ 
One can verify that $(y_i,y_i+b),(y_j,y_j+b), (y_i+b,y_i),(y_j+b,y_j)$ are four distinct solutions of System~\eqref{LBCT} if and only if $ a\in V (i,j)$. Therefore, $\LB_F(a,b,c) \geq 4$ if $a \in V (i,j)$ and $\DDT_F(b,c)=2k$. For $k \geq 4$ and $1 \leq i_1 \neq j_1 \neq i_2 \neq j_2 \leq k$, let $a \in V(i_1, j_1)$ and $a' \in V(i_2, j_2)$. If $a=a'$, then System~\eqref{LBCT} has at least eight solutions, namely, $(y_{i_1},y_{i_1}+b),(y_{j_1},y_{j_1}+b), (y_{i_1}+b,y_{i_1}),(y_{j_1}+b,y_{j_1}),(y_{i_2},y_{i_2}+b),(y_{j_2},y_{j_2}+b), (y_{i_2}+b,y_{i_2}),(y_{j_2}+b,y_{j_2})$; otherwise, it will have at least four solutions. Following a similar argument as in the case of the EBCT, this process must stop somewhere, depending on the value of $\DDT_{F}(b,c)$.

 Again, we can take $\DDT_{F}(b,c)=2k \in \{4r, 4r+2\}$, where $r$ is a non-negative integer. We have already discussed the case $k=1$ in Case 1. For $k=2,3$, one can see that $\LB_F(a,b,c) \in \{0,4\}$. Let us now assume that $k \geq 4$. Let $1\leq \ell \leq r$ be the largest integer for which there are distinct integers $i_1,i_2,\ldots, i_{\ell},j_1,j_2,\ldots, j_{\ell} \in \{1,2, \ldots, k \}$ such that $a \in V(i_1, j_1) \cap V(i_2,j_2) \cap \cdots \cap V(i_{\ell}, j_{\ell})$, then System~\eqref{LBCT} will have $4 \ell$ solutions. This completes the proof for the LBCT.

Finally, we deal with the UBCT  at $(a,b,c)$ of a function $F$ by analyzing System~\eqref{UBCT}. It is immediate that $F(X)+F(X+a)=F(Y)+F(Y+a)=b$ from the equations of System~\eqref{UBCT}. Moreover, $F(X)+F(X+a)=b$ can have zero or $2k$ solutions, where $k\in \left\{1,2,\ldots,\frac{\delta}{2}\right\}$, as $F$ is differentially $\delta$-uniform. Clearly, if $\DDT_F(a,b)=0$, then  System~\eqref{UBCT} has no solutions, and therefore $\UB_F(a,b,c)=0$. 
Let elements of $S(a,b):=\{x_1,x_2,\ldots,x_k,x_1+a,x_2+a,\ldots, x_k+a\}$ be the solution set of $F(X)+F(X+a)=b$. Next, we consider the following cases:

\noindent\textbf{Case 1.} If $a=0, b \neq0$, then $\UB_F(a,b,c)=0$.

\noindent\textbf{Case 2.} Let $a=b=0$, $\Im(F) = \{t \in \F_{2^n} \mid \exists ~u \in \F_{2^n} \mbox{~satisfying~} F(u)= t \}$, and $c + \Im(F) = \{c + t \mid t \in \Im(F)\}$. Then, $\UB_F(0, 0, c) = |F^{-1}(c + \Im(F))|$, the cardinality of the preimage of $c + \Im(F)$.

\noindent\textbf{Case 3.} 
If $b=c$ and $a \neq0$, System~\eqref{UBCT} reduces to \[  \begin{cases} 
  F(X+a)+F(Y+a)=b \\
  F(X)+F(Y)=b \\
  F(X)+F(X+a)=b
  \end{cases}\]
Let $X\in S(a,b)$. Then, $(X,X+a)$ is a solution of the above system.  On the other hand, if $X \notin S(a,b)$, then there does not exist any $Y \in \F_{2^n}$ such that $(X,Y)$ is a solution of the above system. Therefore, $\UB(a,b,c)=\DDT_F(a,b)$.

If $\DDT_F(a,b)=2$, then,  $S(a,b):=\{x_1,x_1+a\}$. Without loss of generality, let $X=x_1$. Since $\{X,X+a,Y,Y+a\}$ be solutions of $F(X)+F(X+a)=b$, and since $a\neq0$, we have that $Y=x_1$ or $Y=x_1+a$. 
If $Y=x_1$ and $c\neq0$, this yields a contradiction (note that, if $c=0$, we get that $(x_1,x_1)$ and $(x_1+a,x_1+a)$ are solutions of the system, and therefore $\UB(a,b,c)=\DDT_F(a,b)$). If $Y=x_1+a$, this yields a contradiction unless $b=c$. Therefore, if $\DDT_F(a,b)=2$, then System~\eqref{UBCT} has a solution only if $c=0$ or $b=c$.

\noindent\textbf{Case 4.} Now assume that $b\neq c,$ $a\neq0$, and $\DDT_F(a,b)=2k$, where $k \in \{2,3,\cdots,\delta/2\}$. For any $X\in S(a,b)$, we have that $(X,X)$ cannot be a solution of System~\eqref{UBCT} unless $c=0$; similarly, $(X,X+a)$ cannot be a solution of System~\eqref{UBCT} unless $c=b$, which is excluded.

Therefore for $X=x_i\in S(a,b)$, if there exist $Y \in S(a,b) \setminus \{x_i,x_i+a\}$ such that $c=F(x_i)+F(Y)=F(x_i+a)+F(Y+a)$, then $(X,Y)$ is a solution of System~\eqref{UBCT}. We now consider the set 
$$W(i,j)= \left\{F(x_i)+F(x_j), F(x_i)+F(x_j)+b ~\mid ~ 1\leq i \neq j\leq k\right\} \subseteq \F_{2^n}^{*}. $$
It is easy to observe that $(x_{i_1},x_{j_1}),(x_{i_1}+a,x_{j_1}+a), (x_{j_1},x_{i_1}),(x_{j_1}+a,x_{i_1}+a)$ are four distinct solutions of System~\eqref{UBCT} if and only if $ c\in W (i_1,j_1)$. It is important to see that for a given $c \in W(i_1,j_1)$, there can be more than four solutions to  System~\eqref{UBCT}, except for these four. 

By following a similar analysis as done in the case of the EBCT and LBCT, depending on the value of the DDT entries, we can compute the UBCT entries of $F$. We can write that $\DDT_F(a,b)=2k \in \{4r,4r+2\}$, where $r\geq 0$ is an integer. If $k=1$, then using Case 1, we are done, and if $k=2,3$, we have $\UB_F(a,b,c) \in \{0,4\}$. Assume that $k \geq 4$. Let $1 \leq \ell \leq r$ be the largest integer for which there are distinct integers $i_1,i_2,\ldots, i_{\ell},j_1,j_2,\ldots, j_{\ell} \in \{1,2, \ldots, k \}$ such that $c \in W(i_1, j_1) \cap W(i_2,j_2) \cap \cdots \cap W(i_{\ell}, j_{\ell})$, then System~\eqref{UBCT} will have $4 \ell$ solutions. This proves the result for the UBCT.
\end{proof}

We now give an example depicting the above Theorem~\ref{deltadiff}. 
\begin{exmp} 
Let $F(X)=X^{11}$, which is a permutation over $\F_{2^6}$. Let $c=g, d=g^{11}$, where $g$ is primitive element of $\F_{2^6}$. Here, $\DDT_F(c,d)=10$ and $S(c,d)=\{z_1=g, z_2=g^{10}, z_3= g^{19}, z_4=g^{22}, z_5=g^{46}, z_1+c=0, z_2+c=g^{28}, z_3+c=g^{55}, z_4+c=g^{43}, z_5+c=g^{37}\}$. If $i_1=1$ and $j_1=3$, then for $(a,b)=\left(z_{i_1}+z_{j_1}, F(z_{i_1})+F(z_{j_1})\right) =(g^{55},g^{38})\in U(i_1,j_1)$, we are interested in computing the largest integer $1 \leq \ell \leq 2$ such that $(a,b) \in U(1,3) \cap  U(i_2,j_2) \cap \cdots  \cap U(i_{\ell},j_{\ell}) $, where $i_2,i_3,\cdots,i_{\ell},j_2,j_3,\cdots,j_{\ell} \in \{2,4,5\}$ are distinct integers. In this example, $\ell$ is indeed $2$. This is because when $i_2=2$ and $j_2=5$, we get $z_{i_1}+z_{j_1}=z_{i_2}+z_{j_2}$ and  $F(z_{i_1})+F(z_{j_1})=F(z_{i_2})+F(z_{j_2})$. Consequently, we have $(a,b) \in U(1,3) \cap U(2,5)$, rendering $\EB_F(a,b,c,d)=4\ell=8$. It also follows that $(a',b')=(z_{i_1}+z_{i_2}+c, F(z_{i_1})+F(z_{i_2})+d )=(g^{19},g^{20})\in U(1,3) \cap U(2,5)$, or equivalently for $a'= g^{19}$ and $b'= g^{20}$, we have $\EB_F(a',b',c,d)=8$.

However, $\DDT_F(b,c)=2$ and $a \neq b$, which gives that $\LB_F(a,b,c)=0$.  Moreover, $\DDT_F(a,b)=10$, and $S(a,b)=\{x_1=g,x_2=g^{10},x_3=g^{13},x_4=g^{28},x_5=g^{55}, x_1+a=g^{19},x_2+a=g^{46},x_3+a=g^{34},x_4+a=g^{37}, x_5+a=0\}$. Then, one can check that $c \neq b$ and  furthermore, $c \not \in W(i,j)$ for all $i \neq j \in \{1,2,3,4\}$, implying that $\UB_F(a,b,c)=0$.
\end{exmp}

Next, we provide a characterization of the EBCT entries of a differentially $\delta$-uniform function $F$ in terms of UBCT and LBCT entries of $F$, and vice versa. 
\begin{thm}\label{GE2LU}
Let $F$ be a function on $\F_{2^n}$. Given any $a,b,c,d\in\F_{2^n}$, the following statements are equivalent:
\begin{itemize}
\item[$(i)$] $x_0\in S(c,d)+a$ is a solution of the system of EBCT$_F(a,b,c,d)$.
\item[$(ii)$] $(x_0,x_0+c)$ is a solution of the system of LBCT$_F(a,c,d)$ and $b=F(x_0)+F(x_0+a)$.
\item[$(iii)$] $(x_0,x_0+a)$ is a solution of the system of UBCT$_F(c,d,b)$ (note that this is only possible if $x_0+a\in S(c,d)$).
\end{itemize}

\end{thm}

\begin{proof}
We have that
 \begin{equation*}
 \EB_F(a, b, c,d) = \left\lvert \left\{ X \in \F_{2^n}  \bigg| 
  \begin{cases} 
(1)\ F(X)+F(X+a)=b\\
(2)\ F(X)+F(X+c)=d\\
(3)\ F(X+a+c)+F(X+a)=d
  \end{cases}
  \right\}\right \rvert .
\end{equation*} 

On the other hand (replacing the input $(a,b,c)$ by $(a,c,d)$),
\begin{equation*}
\label{L}\LB_F(a,c,d) = \left\lvert \left\{ X \in \F_{2^n} \bigg| \exists~ Y \in \F_{2^n}  \mbox{ with }
  \begin{cases} 
 (4)\ X+Y=c \\
(5)\  F(X+a)+F(Y+a)=d \\
(6)\  F(X)+F(Y)=d
  \end{cases}
  \right\}\right \rvert .
\end{equation*} 
 
By Equation (4), $Y=X+c$, and  we can rewrite Equations (5) and (6) as
$F(X+a)+F(X+a+c)=d $ and $ F(X)+F(X+c)=d$. 
Since these are, respectively, Equation (3) and Equation (2), then, it is clear that $x_0$ is a solution of the system of EBCT$_F(a,b,c,d)$ if and only if $(x_0,x_0+c)$ is a solution of the system of LBCT$_F(a,c,d)$ and $b=F(x_0)+F(x_0+a)$. 

Furthermore, (replacing the input  $(a,b,c)$ by $(c,d,b)$),
\begin{equation*}
\UB_F(c, d, b) = \left\lvert \left\{ X \in \F_{2^n} \bigg| \exists~ Y \in \F_{2^n} \mbox{ with }
  \begin{cases} 
(7)\  F(X+c)+F(Y+c)=b \\
(8)\  F(X)+F(Y)=b \\
(9)\  F(X)+F(X+c)=d
  \end{cases}
  \right\}\right \rvert .
\end{equation*} 
Note that Equation (9) is identical to Equation (2).

Taking $Y=X+a$, we obtain from Equations (7) and (8), $F(X+c)+F(X+a+c)=b$, respectively, $F(X)+F(X+a)=b$. 
Note that the second equation   is identical to Equation (1).
Adding the two equations, we obtain:
$F(X)+F(X+c)=F(X+a)+F(X+a+c),$
which, together with Equation (2), renders $F(X+a)+F(X+a+c)=d$, that is, Equation~(3).
From here, we obtain that $x_0$ is a solution of the system of EBCT$_F(a,b,c,d)$ if and only if $(x_0,x_0+a)$ is a solution of the system of UBCT$_F(c,d,b)$. Note that this forces $x_0+a\in S(c,d)$, given Equations (2) and (3).
\end{proof}

\section{Consequences of our results for APN functions}\label{S5}

We first state the EBCT, LBCT and UBCT entries of an APN function $F$ and give these entries in terms of the DDT entries of $F$. It is known that $H(X)=X^{-1}$ over $\F_{2^n}$, where $n$ is odd, is an APN function and the authors in~\cite{EM, MLLZ} have computed its EBCT, LBCT and UBCT entries, which turns out to be a special case of the following result.  

\begin{cor} 
\label{APN}
If $F$ is an APN function over $\F_{2^n}$, then for $a,b,c,d \in \F_{2^n}$, we have 
\begin{align*}
\EB_F(a, b, c,d)& = 
  \begin{cases} 
  \DDT_F(a,b) & \mbox{if}~ c=d=0,\\
  \DDT_F(c,d) & \mbox{if}~~ ac \neq 0, a=c \mbox{ and } b=d; \mbox{ or }a=0, b=0 \mbox{ and } c \neq 0, \\
  0 & \mbox{otherwise},
  \end{cases}\\
\LB_F(a, b, c) &= 
  \begin{cases} 
  \DDT_F(b,c) & \mbox{if}~ b=c=0; \mbox{ or }a=b \mbox{ and } ab \neq 0; \mbox{ or } a=0 \mbox{ and } b \neq 0,\\
  0 & \mbox{otherwise,}
  \end{cases}\\
\UB_F(a, b, c) &= 
  \begin{cases} 
  |F^{-1}(c+\Im(F))|& \mbox{if}~a=0 \mbox{ and } b= 0, \\
  \DDT_F(a,b) & \mbox{if}~ c=b \mbox{ and } a \neq 0, \\
  0 & \mbox{otherwise,}
  \end{cases}
\end{align*} 
where $\Im(F)$ is the image of $F$ and $F^{-1}(\cdot)$ denotes the preimage of the argument.
\end{cor}
\begin{proof}
The proof of this corollary is straightforward from Theorem~\ref{deltadiff}, by setting $\delta=2$.
\end{proof}

\begin{cor} 
For any $a,b,c,d\in\F_{2^n}$, $(\EB_F(a,b,c,d))^2\leq\LB_F(a,c,d)\cdot\UB_F(c,d,b)$. Furthermore, the equality is met for every $a,b,c,d\in\F_{2^n}^*$ if and only if $F$ is~APN. 
\end{cor}
\begin{proof} The inequality follows directly from Theorem~\ref{GE2LU}. Suppose that $F$ is APN. Then, by Corollary \ref{APN}, $(\EB_F(a,b,c,d))^2=\LB_F(a,c,d)\cdot\UB_F(c,d,b)$. Suppose now that $F$ is not APN. Then, for some $c,d\in\F_{2^n}^*$, there exist at least four distinct solutions of the equation $F(X+c)+F(X)=d$, which we denote by $z_1,z_2,z_1+c,z_2+c$. Let now $a=z_1+z_2,b=F(z_1)+F(z_2)+d$. Then, since $a \neq c$ (otherwise $z_2=z_1+c$), $(\EB_F(a,b,c,d))^2<\LB_F(a,c,d)\cdot\UB_F(c,d,b)$. 
\end{proof}

\section{Consequences of our results for some well-known functions}\label{S6}

In this section, we will give the EBCT, LBCT and UBCT entries of the Gold, Kasami and Bracken-Leander functions for those parameters such that the functions are differentially 4-uniform permutations; in the case of the Gold function, we also state and prove the EBCT, LBCT and UBCT entries for any parameters. We can also give a very short proof for the EBCT, LBCT and UBCT entries of the inverse function for $n$ even, which is the main result in \cite{EM, MLLZ} (for $n$ odd, see the previous section).

For clarity, we first state the results we will use to compute the EBCT, LBCT and UBCT entries of differentially $4$-uniform functions using Theorem~\ref{deltadiff} and Theorem~\ref{GE2LU}, which we will need for the functions in this section. 
\begin{cor}\label{4diff} Let $F$ be a differentially $4$-uniform function over $\F_{2^n}$. Then,
 for $a,b,c,d \in \F_{2^n}$, we have
 \begin{equation*}\EB_{F}(a, b, c,d) = 
  \begin{cases} 
\DDT_F(a,b) & \mbox{if } c = d =0,\\
  \DDT_{F}(c,d) & \mbox{if}~ ac \neq 0, a=c \mbox{ and } b=d; \mbox{ or }a=0, b=0 \mbox{ and } c \neq 0; \mbox{ or }\\
&~ \DDT_{F}(c,d)=4, a=z_1+z_2 \mbox{ and }  b=F(z_1)+F(z_2); \mbox{ or }\\
&~ \DDT_{F}(c,d)=4,a=z_1+z_2+c \mbox{ and }b=F(z_2)+F(z_2+c),\\
  0 & \mbox{otherwise},
  \end{cases}
\end{equation*} 
where $z_1,z_1+c,z_2,z_2+c$ are the four solutions of the equation $F(X+c)+F(X)=d$, if $\DDT_F(c,d)=4$. Next,

\begin{equation*}\LB_{F}(a, b, c) = 
  \begin{cases} 
  \DDT_{F}(b,c) & \mbox{if}~b=c=0; \mbox{ or } a=0 \mbox{ and } b \neq 0; \mbox{ or }\\ 
  & ~\DDT_{F}(b,c)=2, a=b \mbox{ and }  ab \neq 0; \mbox{ or }\\
&~ \DDT_{F}(b,c)=4, ab \neq 0  \mbox{ and }  a\in \{ b, y_1+y_2, y_1+y_2+b\}, \\
  0 & \mbox{otherwise},
  \end{cases}
\end{equation*} 
where $y_1,y_1+b,y_2,y_2+b$ are the four solutions of the equation $F(X+b)+F(X)=c$, if $\DDT_F(b,c)=4$. Further,
\begin{equation*}\UB_{F}(a, b, c) = 
  \begin{cases} 
  |F^{-1}(c+\Im(F))|& \mbox{if}~a= b= 0, \\
  \DDT_{F}(a,b) & ~\DDT_{F}(a,b)=2, a\neq 0 \mbox{ and }  c=b; \mbox{ or } \\
& ~\DDT_{F}(a,b)=4, a\neq 0 \mbox{ and }\\
& ~c\in \{ b, F(x_1)+F(x_2),F(x_1)+F(x_2)+b\}, \\
  0 & \mbox{otherwise},
  \end{cases}
\end{equation*} 
where $x_1,x_1+a,x_2,x_2+a$ are the four solutions of the equation $F(X+a)+F(X)=b$, if $\DDT_F(a,b)=4$, $\Im(F)$ is the image of $F$ and $F^{-1}(\cdot)$ denotes the preimage of the argument.

\end{cor}
\begin{proof} 
The proof directly follows from Theorem~\ref{deltadiff} by taking $\delta=4$.
\end{proof}

Here, we give an alternative characterization of EBCT of a differentially $4$-uniform function~$F$ using LBCT and UBCT of $F$.
\begin{cor}\label{E2LU}
 Let $F$ be a differentially $4$-uniform function over $\F_{2^n}$.  Let  $A=\left\{(a,c,d) \bigm|  \exists\, b, \right. \\ \left. \text{with~} \EB_{F}(a, b, c,d) = 4 \right\}$ and $B=\left\{(c,d,b)\bigm|  \exists\, a \text{~with~} \EB_{F}(a, b, c,d) = 4\right\}$.  Let $a,b,c,d \in \F_{2^n}^{*}$. Then, $\EB_{F}(a, b, c,d) = 2$ if and only if $\LB_{F}(a, c, d) =\UB_{F}(c, d, b)=2$. Furthermore, $(a,c,d)\in A$ if and only if $\LB_{F}(a, c, d) =4$, and $(c,d,b)\in B$ if and only if $\UB_{F}(c, d, b)=4$.
\end{cor}
\begin{proof}
Let $F$ be differentially $4$-uniform. Then, since $cd\neq0$, either $\DDT_F(c,d)=0$, $\DDT_F(c,d)=2$ or $\DDT_F(c,d)=4$. One can now proof the result directly from Theorem~\ref{GE2LU}.
 \end{proof}

We will here use Corollaries~\textup{\ref{4diff}}--~\textup{\ref{E2LU}}  to compute the concrete values for three infinite classes of differentially 4-uniform power permutations over $\F_{2^n}$ (Table ~\ref{Table1}) having the best known nonlinearity. For the Gold function, using Theorem \ref{deltadiff},  we will state and prove the result for general parameters, including those for which the Gold function is 4-differential uniform.

\begin{table}[hbt]
\caption{Differentially 4-uniform permutations $X^d$ over $\F_{2^n}$.}
\label{Table1}
\begin{center}
\begin{tabular}{|c|c|c|c|} 
 \hline
 Family & $d$ & condition &  LBCT/UBCT\\
 \hline
 Inverse & $2^n-2$ & $n=2k, k>1$ &  \textup{\cite{EM, MLLZ}}\\
 \hline
 Gold & $2^s+1$ & $n=2k$, $k$ odd, $\gcd(s,n)=2$ &  This paper \tablefootnote{Shown for general parameters.}\\
 \hline
 Kasami & $2^{2s}-2^s+1$ & $n=2k$, $k$ odd, $\gcd(s,n)=2$ &  This paper\\
 \hline
 Bracken-Leander & $2^{2s}+2^s+1$ & $n=4s$, $s$ odd &  This paper\\
 \hline
\end{tabular}
\end{center}
\end{table}

We first consider the Gold function $F_1(X)=X^{2^s+1}$, a differentially $2^{t}$-uniform function over $\F_{2^n}$, where $t=\gcd(s,n)$ and compute its EBCT, LBCT and UBCT entries. We will use the relative trace in the following way. For a general $s$, we do not have in general that $\F_{2^s}\subseteq\F_{2^n}$. However, we can naturally embed the elements of $\F_{2^n}$ in $\F_{2^{sm}}$, since $m=\frac{n}{\gcd(s,n)}$. Then, $\sum_{i=0}^{m-1}\alpha^{2^{si}}=\Tr_s^{sm}(\alpha)$. 
  
\begin{cor}
\label{ggold}
Let $F_1(X)=X^{2^s+1}$ be the Gold function on $\F_{2^n}, 1\leq s<n$. Let  $t=\gcd(s,n)$, and $m=\frac{n}{t}$. 
Then for $a,b,c,d \in \F_{2^n}$, we have
 \allowdisplaybreaks
 \begin{align*}
\EB_{F_1}(a, b, c,d) &= 
  \begin{cases} 
  2^n & \mbox{if } a=b=c=d=0,\\
  2^{t} & \mbox{if }\Tr_s^{sm}\left(\frac{b}{a^{2^s+1}}\right)=\Tr_{1}^{m}(1), c=d=0 \mbox{ and } a \neq 0 ; \mbox{ or }\\
   &~ \Tr_s^{sm}\left(\frac{d}{c^{2^s+1}}\right)=\Tr_{1}^{m}(1), a=b=0 \mbox{ and } c \neq 0 ;\mbox{ or }\\
  &~ \Tr_s^{sm}\left(\frac{d}{c^{2^s+1}}\right)=\Tr_{1}^{m}(1), a=c, b=d \mbox{ and } ac \neq 0; \mbox{ or }\\
&~ \Tr_s^{sm}\left(\frac{d}{c^{2^s+1}}\right)=\Tr_{1}^{m}(1),  a=uc, ac \neq 0 \mbox{ and } \\
&~ b=(u+u^2)c^{2^s+1}+u d, \mbox{ where }  u\in\F_{2^{t}}^*\setminus \{1\}, \\
0 & \mbox{otherwise,}
\end{cases}\\
\LB_{F_1}(a, b, c) &= 
\begin{cases}
2^n & \mbox{if } b=c=0,\\
2^{t}&\mbox{if } \Tr_s^{sm}\left(\frac{c}{b^{2^s+1}}\right)=\Tr_{1}^{m}(1), a=0 \mbox{ and } b \neq 0; \mbox{ or }  \\
&~ \Tr_s^{sm}\left(\frac{c}{b^{2^s+1}}\right)=\Tr_{1}^{m}(1),  a\in b\F_{2^{t}}^* \mbox{ and } ab \neq 0,\\
  0 & \mbox{otherwise,}
\end{cases}\\
\UB_{F_1}(a, b, c) &= 
\begin{cases}
|F_{1}^{-1}(c+\Im(F_1))| & \mbox{if }~a=b= 0, \\
2^{t} & \mbox{if }a\neq 0, \Tr_s^{sm}\left(\frac{b}{a^{2^s+1}}\right)=\Tr_{1}^{m}(1)  \mbox{ and }\\
&~ c\in \{b, (u+u^2)a^{2^s+1}+u b, u\in\F_{2^{t}}^*\setminus \{1\}\}, \\
  0 & \mbox{otherwise,}~
\end{cases}\\
\end{align*} 
where $\Im(F_1)$ is the image of $F_1$ and $F_{1}^{-1}(\cdot)$ denotes the preimage of the argument. In particular when $F_1$ is a permutation or equivalently when $m$ is odd, then $|F_{1}^{-1}(c+\Im(F_1))|=2^n$.
\end{cor} 

\begin{proof} 
By Theorem~\ref{deltadiff}, we have 
\begin{equation*}
\EB_{F_1}(a, b, c,d) = 
  \begin{cases}
\DDT_{F_1}(a,b) & \mbox{if } c = d =0,\\
  \DDT_{F_1}(c,d) & \mbox{if}~ ac \neq 0, a=c \mbox{ and } b=d; \mbox{ or }a=b=0 \mbox{ and } c \neq 0,\\
4\ell &\mbox{if } a\neq c, ac \neq 0 \mbox{ and } \DDT_{F_1}(c,d)=2k=4r \mbox{~or~} 4r+2,
\\
&~\mbox{where}~r>0\mbox{~is an integer,}\\
&~1 \leq \ell \leq r~\mbox{ is the largest integer such that} \\
&~(a,b) \in U(i_1, j_1) \cap U(i_2,j_2) \cap \cdots \cap U(i_{\ell}, j_{\ell}),\\
& ~ 1\leq i_1,i_2,\ldots, i_{\ell},j_1,j_2,\ldots, j_{\ell}\leq k~\mbox{are distinct integers,}\\
0 & \mbox{otherwise},
  \end{cases}
\end{equation*} 
where $U(i,j)= \{\left(z_i+z_j, F_1(z_i)+F_1(z_j)\right), \left(z_i+z_j+c, F_1(z_i)+F_1(z_j)+d\right) ~\mid ~ 1\leq i \neq j\leq k\} \subseteq \F_{2^n}^{*} \times \F_{2^n}^{*}$ and the set $S(c,d)=\{z_1,z_1+c,z_2,z_2+c, \ldots,z_{k},z_{k}+c\}$ consists of $2k$ solutions of the equation $F_1(X+c)+F_1(X)=d$ when $\DDT_{F_1}(c,d)=2k$.

It is known from~\cite[Lemma 4]{BCC} and \cite{CH} that $ X^{2^s+1}+(X+c)^{2^s+1}=d$ has  $2^{t}$ distinct solutions if and only if $\sum_{i=0}^{m-1}D^{2^{si}}=m \equiv \Tr_{1}^{m}(1)$ (mod 2), where $D=\frac{d}{c^{2^s+1}}$, with $t=\gcd(s,n)$ and $m=\frac{n}{t}$, and zero solutions otherwise. Moreover, the solutions are of the form $X+cu$, where $u\in\F_{2^{t}}$.

If $t=1$, we get that $F_1$ is APN, and using Corollary~\ref{APN}, we are done. Next, let us suppose that $t>1$. Here, we focus exclusively on the case where $a \neq c$ and $ac \neq 0$, as all other cases are straightforward. Let $\DDT_{F_1}(c,d)=2k=2^{t}=4r$, where $r$ is a positive integer. Let $z_{i_1}=X+u_{i_1}c$ and $z_{j_1}=X+u_{j_1}c$, where $u_{i_1} \neq u_{j_1} \in  \F_{2^t}$ and $u_{i_1}+u_{j_1} \neq c$. Then, for $(a,b)=(z_{i_1}+z_{j_1},F_1(z_{i_1})+F_1(z_{j_1}))=((u_{i_1}+u_{j_1})c,F_1(X+u_{i_1})+F_1(X+u_{j_1}))$, System~\eqref{EBCT} will have at least four solutions, namely, $\{X+u_{i_1}c,X+u_{j_1} c,X+u_{i_1}c +c,X+u_{j_1}c+c\}$, where $X$ satisfies $X^{2^s+1}+(X+c)^{2^s+1}=d$.

Next, we are required to examine if System~\eqref{EBCT} can have any other solution from the set $S(c,d)$ except for these four, corresponding to the pair $(a,b)$ defined above. Using Theorem ~\ref{deltadiff}, we know that $z_{i_2}=X+u_{i_2} c \in S(c,d) \setminus \{ z_{i_1},z_{j_1},z_{i_1}+c,z_{j_1}+c\}$, where $u_{i_2} \in \F_{2^{t}}$, is a solution of System~\eqref{EBCT} if there exist $z_{j_2}\in S(c,d) \setminus \{ z_{i_1},z_{j_1},z_{i_1}+c,z_{j_1}+c,z_{i_2},z_{i_2}+c\}$ such that $a=z_{i_1}+z_{j_1}=z_{i_2}+z_{j_2}$ and $b=F_1(z_{i_1})+F_1(z_{j_1})=F_1(z_{i_2})+F_1(z_{j_2})$. Note that if we choose $z_{j_2}=X+(u_{i_1}+u_{j_1}+u_{i_2})c$, then it is easy to see that $a=z_{i_1}+z_{j_1}=u_{i_1}+u_{j_1}=z_{i_2}+z_{j_2}$ and moreover, the following computations
\allowdisplaybreaks
\begin{align*}
 F_1(z_{i_1})+F_1(z_{j_1}) & =(X+u_{i_1}c)^{2^s+1}+(X+u_{j_1}c)^{2^s+1}\\
 & = ((u_{i_1}+u_{j_1})c)^{2^s}X+(u_{i_1}+u_{j_1})cX^{2^s}+(u_{i_1}^{2^s+1}+u_{j_1}^{2^s+1})c^{2^s+1}\\
 & = ((u_{i_1}+u_{j_1})c)^{2^s}X+(u_{i_1}+u_{j_1})cX^{2^s}+(u_{i_1}+u_{j_1})^{2^s+1}c^{2^s+1}\\
 & \qquad +(u_{i_1}^{2^s}u_{j_1}+u_{j_1}^{2^s}u_{i_1})c^{2^s+1}\\
 & = (u_{i_1}+u_{j_1}) (c^{2^s}X+cX^{2^s})+ (u_{i_1}+u_{j_1})^2 c^{2^s+1}\\
 & = (u_{i_1}+u_{j_1}) (c^{2^s+1}+d)+ (u_{i_1}+u_{j_1})^2 c^{2^s+1}\\
 &= ((u_{i_1}+u_{j_1})+(u_{i_1}+u_{j_1})^2 ) c^{2^s+1}+ (u_{i_1}+u_{j_1}) d,\\
F_1(z_{i_2})+F_1(z_{j_2}) & =(X+u_{i_2}c)^{2^s+1}+(X+(u_{i_1}+u_{j_1}+u_{i_2})c)^{2^s+1}\\
 & = ((u_{i_1}+u_{j_1})c)^{2^s}X+(u_{i_1}+u_{j_1})cX^{2^s}+(u_{i_1}+u_{j_1})^{2^s+1}c^{2^s+1}\\
 & \quad +(u_{i_2}^{2^s}(u_{i_1}+u_{j_1})+(u_{i_1}+u_{j_1})^{2^s}u_{i_2})c^{2^s+1}\\
 & = (u_{i_1}+u_{j_1}) (c^{2^s}X+cX^{2^s})+ (u_{i_1}+u_{j_1})^2 c^{2^s+1},\\
 &= ((u_{i_1}+u_{j_1})+(u_{i_1}+u_{j_1})^2 ) c^{2^s+1}+ (u_{i_1}+u_{j_1}) d,
\end{align*}
show that $ b =F_1(z_{i_1})+F_1(z_{j_1})=F_1(z_{i_2})+F_1(z_{j_2})$.
We have used here that, since $t|s$, if $u\in\F_{2^{t}}$ then $u^{2^s}=u$, and that $d=(X+u_{i_1} c+c)^{2^s+1}+(X+u_{i_1} c)^{2^s+1}=c^{2^s}X+cX^{2^s}+c^{2^s+1}$. Hence, for $(a,b)=(z_{i_1}+z_{j_1},F_1(z_{i_1})+F_1(z_{j_1}))$, we get $z_{i_2}$ as a solution of System~\eqref{EBCT}. As $z_{i_2} \in S(c,d) \setminus \{ z_{i_1},z_{j_1},z_{i_1}+c,z_{j_1}+c\}$ was chosen arbitrarily, so every element in the set $S(c,d)$ will be a solution of System~\eqref{EBCT} for $(a,b)=(z_{i_1}+z_{j_1},F_1(z_{i_1})+F_1(z_{j_1}))$. Therefore, for $\DDT_{F_1}(c,d)=2k$, it turns out that $\ell =r$ is the largest integer such that $(a,b) \in U(i_1, j_1) \cap U(i_2,j_2) \cap \cdots \cap U(i_{\ell}, j_{\ell}),$ where $1\leq i_1,i_2,\ldots, i_{\ell},j_1,j_2,\ldots, j_{\ell}\leq k$ are distinct integers; or, equivalently, for $\DDT_{F_1}(c,d)=2k$ and $(a,b) \in \{(uc, (u+u^2)c^{2^s+1}+ud) \mid u \in \F_{2^t}^{*} \setminus \{1\}\}$, $\EB_{F_1}(a,b,c,d)=4 \ell =\DDT_{F_1}(c,d)$.

We now compute LBCT and UBCT entries of $F_1$. If $t=1$, then $F_1$ is APN, and using Corollary~\ref{APN}, we are done. Next, let us suppose that $t>1$, then we will compute the LBCT and UBCT entries via (and using the notations and expressions of LBCT and UBCT) Theorem~\ref{deltadiff}. 

 It is known from Theorem~\ref{deltadiff} that if $a=b, ab \neq 0$ or $a=0, b \neq 0$, then $\LB_{F_1}(a,b,c)=\DDT_{F_1}(b,c)=2^t$ if and only if $\Tr_s^{sm}\left(\frac{c}{b^{2^s+1}}\right)=\Tr_{1}^{m}(1)$; otherwise, it is zero. Now, let $a \neq b, ab \neq 0$ and $\DDT_{F_1}(b,c)=2k=2^{t}=4r$, where $r$ is a positive integer. Consider $S(b,c)=\{y_1,y_1+b,y_2,y_2+b, \ldots,y_{k},y_{k}+b\}$ to be the solution set of the equation $F_1(X+b)+F_1(X)=c$. Suppose that $y_{i_1}=X+v_{i_1}b$ and $y_{j_1}=X+v_{j_1}b$, where $X$ satisfies $X^{2^s+1}+(X+b)^{2^s+1}=c$, $v_{i_1} \neq v_{j_1} \in  \F_{2^t}$ and $v_{i_1}+v_{j_1} \neq b$. Then, for $a=y_{i_1}+y_{j_1}=(v_{i_1}+v_{j_1})b$, System~\eqref{LBCT} has at least four solutions, namely, $\{(X+v_{i_1}b,X+v_{i_1}b+b),(X+v_{j_1} b,X+v_{j_1}b+b),(X+v_{i_1}b +b,X+v_{i_1}b), (X+v_{j_1}b+b,X+v_{j_1}b)\}$. We now need to investigate that whether  System~\eqref{LBCT} has more than these four soultions, corresponding to $a=(v_{i_1}+v_{j_1})b$. From Theorem~\ref{deltadiff}, it follows that $(y_{i_2},y_{i_2}+b)$, where $y_{i_2}=X+v_{i_2}b \in S(b,c)\setminus \{y_{i_1},y_{j_1},y_{i_1}+b,y_{j_1}+b\}$ and $v_{i_2} \in \F_{2^t}$, is a solution of System~\eqref{LBCT} if and only if there exists $y_{j_2}=X+v_{j_2}b  \in S(b,c)\setminus \{y_{i_1},y_{j_1},y_{i_1}+b,y_{j_1}+b,y_{i_2},y_{i_2}+b\}$, $v_{j_2} \in \F_{2^t}$, such that $a=y_{i_2}+y_{j_2}$. Clearly, one can choose $y_{j_2}=X+(v_{i_1}+v_{i_2}+v_{j_2})$ to ensure $a=y_{i_1}+y_{j_1}=y_{i_2}+y_{j_2}$. Thus, we obtain another set of four solutions $\{(X+v_{i_2}b,X+v_{i_2}b+b),(X+v_{j_2} b,X+v_{j_2}b+b),(X+v_{i_2}b +b,X+v_{i_2}b), (X+v_{j_2}b+b,X+v_{j_2}b)\}$ for System~\eqref{LBCT} when $a=y_{i_1}+y_{j_1}$. Since $y_{i_2}\in S(b,c)\setminus \{y_{i_1},y_{j_1},y_{i_1}+b,y_{j_1}+b\}$ was choosen arbitrarily, it renders that every element in the set $S(b,c)$ is a solution of System~\eqref{LBCT} for $a=y_{i_1}+y_{j_1}$. Thus, $\LB_{F_1}(a,b,c)=\DDT_{F_1}(b,c)$.

 We next compute the UBCT entries for $F_1$. It is evident from Theorem~\ref{deltadiff} that
 if $a=b=0$, then $\UB_{F_1}(a,b,c)=|F_{1}^{-1}(c+\Im(F_1))|$ where $\Im(F_1)$ is the image of $F_1$ and $F_{1}^{-1}(\cdot)$ denotes the preimage of the argument. Moreover, if $b=c$ and $a \neq 0$ then $\UB_{F_1}(a,b,c)=\DDT_{F_1}(a,b)=2^t$ if and only if $\Tr_s^{sm}\left(\frac{b}{a^{2^s+1}}\right)=\Tr_{1}^{m}(1)$; otherwise, it is zero. Therefore, assume that $a \neq 0$ and $b \neq c$ and $\DDT_{F_1}(a,b)=2k=2^{t}=4r$, where $r$ is a positive integer. Define $S(a,b)=\{x_1,x_1+a,x_2,x_2+a, \ldots,x_{k},x_{k}+a\}$ as the solution set of the equation $F_1(X+a)+F_1(X)=b$. Suppose $x_{i_1}=X+w_{i_1}a$ and $x_{j_1}=X+w_{j_1}a$, where $w_{i_1} \neq w_{j_1} \in  \F_{2^t}$, $w_{i_1}+w_{j_1} \neq a$ and $X^{2^s+1}+(X+a)^{2^s+1}=b$. Then, for $c=F_1(x_{i_1})+F_1(x_{j_1})=((w_{i_1}+w_{j_1})+(w_{i_1}+w_{j_1})^2 ) a^{2^s+1}+ (w_{i_1}+w_{j_1}) b$, System~\eqref{UBCT} will have at least four solutions, namely, $\{(X+w_{i_1}a,X+w_{j_1}a),(X+w_{i_1}a+a,X+w_{j_1}a+a),(X+w_{j_1}a,X+w_{i_1}a), (X+w_{j_1}a+a,X+w_{i_1}a+a)\}$. Using the similar argument as in the case of EBCT, one can ensure that every element in the set $S(a,b)$ is a solution of System~\eqref{UBCT} for $c=F_1(x_{i_1})+F_1(x_{j_1})=((w_{i_1}+w_{j_1})+(w_{i_1}+w_{j_1})^2 ) a^{2^s+1}+ (w_{i_1}+w_{j_1}) b$. Hence, we conclude that $\UB_{F_1}(a,b,c)=\DDT_{F_1}(a,b)$.
\end{proof}

We will derive here the $\FB$ for the Gold function for any parameters $a,b$ (also found in \cite{EM22}, but our proof is significantly shorter). As known, for monomials, one can take $b=1$.
\begin{cor} Let $F_1(X)=X^{2^s+1}$, with $\gcd(s,n)=t$ and $m=\frac{n}{t}$, where $a \in \F_{2^n}^{*}$. Then,
 \allowdisplaybreaks
 \begin{align*}
\FB_{F_1}(a, 1) &= 
\begin{cases}
2^n&\mbox{ if } a\in \F_{2^{t}}^{*},\\
  0 & \mbox{otherwise.}~
\end{cases}
\end{align*} 
\end{cor}
\begin{proof}
Observe that
 \allowdisplaybreaks\begin{align*}
\LB_{F_1}(a, 1, c) &= 
\begin{cases}
2^{t} &\mbox{ if }\sum_{i=0}^{m-1}c^{2^{si}}=1 \mbox{ and }  a\in \F_{2^{t}}^{*}, \\
  0 & \mbox{otherwise}.
\end{cases}
\end{align*} 
This implies that, if $a\not \in \F_{2^{t}}^{*}$, then $\FB_{F_1}(a, 1)=0$. Let us now assume that $a\in \F_{2^{t}}^{*}$.  Note that the condition $\sum_{i=0}^{m-1}c^{2^{si}}=1 $ is equivalent to $\DDT_{F_1}(1,c)=2^{t}$. Denoting as is usual $\omega_i=\left|\{c\bigm| \DDT_{F_1}(1,c)=i\}\right|$, $\displaystyle \FB_{F_1}(a, 1)=\sum_{c\in\F_{2^n}}\LB_{F_1}(a, 1, c)=2^{t}\omega_{2^{t}}=\sum_{i=0}^{2^{t}} i\omega_i=2^n$. This proves the claim.
\end{proof}

\begin{rmk}
 \label{independent}
Let $F$ be a differentially $4$-uniform function over $\F_{2^n}$ such that $\DDT_F(b,c)=4$ or $\DDT_F(b,c)=0$ for any $b,c\in\F_{2^n}$. For any $b,c$ such that $\DDT_F(b,c)=4$, denote by $\{y_1,y_1+b,y_2,y_2+b\}$ the four distinct solutions of the equation $F(X+b)+F(X)=c$. If $y_i+y_j$, $i\neq j$, is independent of $c$, then by using a similar argument as in the above corollary, for $a\in\{y_i+y_j,y_i+y_j+b\}$, we obtain
 \allowdisplaybreaks\begin{align*}
\FB_{F}(a, b) &= 
\begin{cases}
2^n&\mbox{ if } a\in \{ b, y_i+y_j, y_i+y_j+b\}\\
  0 & \mbox{otherwise.}
\end{cases}
\end{align*}
\end{rmk}

\begin{exmp}
  In Table~\textup{\ref{Table2}}, Corollary~\textup{\ref{ggold}} is illustrated by giving explicit values of the $\EB$, $\LB$ and $\UB$ entries of the Gold function $X^{2^s+1}$ over $\F_{2^n}^*=\langle g\rangle$, where $g$ is a primitive element of  $\F_{2^n}$, for some small values of~$n$.
\end{exmp}
\begin{table}[hbt]
\caption{EBCT, LBCT and UBCT entries for the Gold function $X^{2^s+1}$ over $\F_{2^{n}}$}
\label{Table2}
\begin{center}
\scalebox{0.8}{
\begin{tabular}{|c|c|c|c|c|c|c|c|c|c|} 
 \hline
 $n$ & $s$ & $a$ & $b$ & $c$ & $d$ &  $\omega$ & $\EB_{F_1}(a,b,c,d)$ &$\LB_{F_1}(a,b,c)$ & $\UB_{F_1}(a,b,c)$\\
 \hline
 $6$ & $2$ & $g^{44}$ & $g^{23}$ & $g^{16}$  & for any $d$  &  $g^{21}$ & $0$ & $4$ & $0$ \\
 \hline
 $6$ & $2$ & $g^{2}$ & $g^{26}$ & $g^{53}$  &  for any $d$  & $g^{21}$ & $0$ & $0$ & $4$\\
 \hline
$10$ & $6$ & $g^{351}$ & $g^{692}$ & $g^{2}$  & for any $d$ &  $g^{341} $ & $0$ & $0$ & $0$\\
 \hline
 $10$ & $4$ & $g^{2}$ & $g^{359}$ & $g^{11}$  & for any $d$ & $g^{682}$ & $0$ & $0$ & $4$\\
\hline
$6$ & $2$ & $g^{44}$ & $g^{8}$ & $g^{23}$  & $g^{16}$  &  $g^{21}$ & $4$ & $0$ & $0$ \\
\hline
$10$ & $4$ & $g^{684}$ & $g^{11}$ & $g^{2}$  & $g^{359}$ &  $g^{682}$ & $4$ & $0$ & $0$\\
 \hline
\end{tabular}
}
\end{center}
\end{table}

The Kasami function is a well-known nonlinear function over $\F_{2^n}$, given by ${F_2}(X)=X^{2^{2s}-2^s+1}$. Hertel and Pott~\cite{HP} showed that for $\gcd(s,n)=2$, where $n=2t$,  $t$ is odd and $3\nmid t$, it is a differentially $4$-uniform permutation. We next investigate the EBCT, LBCT and UBCT entries of $F_2$ over $\F_{2^n}$. 

\begin{cor}\label{Kasami}
Let $F_2(X)=X^{2^{2s}-2^s+1}$ on $\F_{2^n}$, where $\gcd(s,n)=2$, $n=2t$ and $t$ is odd and $3\nmid t$. Then for $a,b,c,d \in \F_{2^n}$, we have
 \allowdisplaybreaks\begin{align*}
\EB_{F_{2}}(a, b, c,d) &= 
  \begin{cases} 
  2^n & \mbox{if } a=b=c=d=0, \\
  4 & \mbox{if } \DDT_{F_{2}}(a,b)=4, a\neq 0 \mbox{ and } c=d=0; \mbox{ or}\\ 
   &~ \DDT_{F_2}(c,d)=4, a=b=0 \mbox{ and } c \neq 0;\mbox{ or}\\
   &~  \DDT_{F_2}(c,d)=4,  ac \neq 0, a=c \mbox{ and }b=d; \mbox{ or}\\
&~ \DDT_{F_{2}}(c,d)=4,  a=\omega c+\alpha^{2^{s}+1},\mbox{ and } b=\omega d+\alpha^{2^{3s}+1}; \mbox{ or}\\
&~ \DDT_{F_{2}}(c,d)=4,  a=\omega^2 c+\alpha^{2^{s}+1},\mbox{ and } b=\omega^2 d+\alpha^{2^{3s}+1},\\
  0 & \mbox{otherwise,}
  \end{cases}
  \end{align*} 
for $\alpha \in \F_{2^n}^{*}$ with $c^{2^{2s}}+c^{2^s}\alpha^{2^{3s}-2^s}+c\alpha^{2^{3s}+2^{2s}-2^s-1}+d\alpha^{2^{2s}-1}=0$ and $
\Tr_s^{st}\left(1+\dfrac{c}{\alpha^{2^s+1}} \right)=0$,
 \allowdisplaybreaks\begin{align*}
\LB_{F_2}(a, b, c) &= 
\begin{cases}
2^n &\mbox{if } b=c=0,\\
4&\mbox{if }\DDT_{F_2}(b,c)=4, a=0 \mbox{ and }  b \neq 0; \mbox{ or}\\
&~ \DDT_{F_2}(b,c)=4, ab\neq 0 \mbox{ and }  a\in \left\{ b, b\omega +\alpha^{2^s+1}, b\omega^2+\alpha^{2^s+1}\right\}, \\
  0 & \mbox{otherwise,}
  \end{cases}
  \end{align*} 
for $\alpha \in \F_{2^n}^{*}$ with $b^{2^{2s}}+b^{2^s}\alpha^{2^{3s}-2^s}+b\alpha^{2^{3s}+2^{2s}-2^s-1}+c\alpha^{2^{2s}-1}=0$ and $\Tr_s^{st}\left(1+\dfrac{b}{\alpha^{2^s+1}} \right)=0$,
 \allowdisplaybreaks\begin{align*}
\UB_{F_2}(a, b, c) &= 
\begin{cases}
2^n &\mbox{if } a=b=0,\\
4 & \mbox{if}~\DDT_{F_2}(a,b)=4, a\neq 0 \mbox{ and }  c\in \left\{  b,  b\omega+\alpha^{2^{3s}+1},  b\omega^2+\alpha^{2^{3s}+1}\right\}, \\
  0 & \mbox{otherwise,}
\end{cases}
\end{align*} 
for $\alpha \in \F_{2^n}^{*}$ satisfying $a^{2^{2s}}+a^{2^s}\alpha^{2^{3s}-2^s}+a\alpha^{2^{3s}+2^{2s}-2^s-1}+b\alpha^{2^{2s}-1}=0$, 
$\Tr_s^{st}\left(1+\dfrac{a}{\alpha^{2^s+1}} \right)=0$
and $\omega$ is the primitive cube root of unity. 
\end{cor}
\begin{proof}
Since $F_2$ is a permutation, Lemma ~\ref{L0001} can be referred for determining EBCT, LBCT and UBCT entries of $F_2$ in cases where $abcd=0$. Thus, in what follows, we will assume that $abcd \neq 0$. By Corollary \ref{4diff}, for $a,b,c,d \in \F_{2^n}^{*}$, $\EB_{F_{2}}(a, b, c,d) $ equals
\begin{equation*} 
  \begin{cases} 
\DDT_{F_{2}}(c,d) & \mbox{if }  a=c \mbox{ and }b=d; \mbox{ or}\\
&~ \DDT_{F_{2}}(c,d)=4,  a=z_1+z_2 \mbox{ and } b=F_{2}(z_1)+F_{2}(z_2); \mbox{ or}\\
&~\DDT_{F_{2}}(c,d)=4, a=z_1+z_2+c \mbox{ and } b=F_{2}(z_1)+F_{2}(z_2+c),\\
  0 & \mbox{otherwise,}
  \end{cases}
\end{equation*} 
where $z_1,z_1+c,z_2,z_2+c$ are the four solutions of the equation $F_{2}(X+c)+F_{2}(X)=d$ if $\DDT_{F_{2}}(c,d)=4$.  We write $F_{2}(X+c)+F_{2}(X)=d$ as the system 
\begin{equation}\label{k0}
  \begin{cases} 
  Y+X=c, \\
  F_2(Y)+F_2(X)=d.
  \end{cases}
\end{equation}
Now by using the argument given in~\cite{CKM}, and the fact that $\gcd(2^s+1,2^n-1)=1$, substituting $X=u^{2^s+1}, Y=v^{2^s+1}=(u+\alpha)^{2^s+1}$, for some $\alpha \in \F_{2^n}^{*}$ such that $v=u +\alpha$ and simplifying System~\eqref{k0}, we have
\begin{equation*}
  \begin{cases} 
  \left(\frac{u}{\alpha}\right)^{2^s}+\frac{u}{\alpha}+1+\frac{c}{\alpha^{2^s+1}}=0, \\
  \left(\frac{u}{\alpha}\right)^{2^{3s}}+\frac{u}{\alpha}+1+\frac{d}{\alpha^{2^{3s}+1}}=0,
  \end{cases}
\end{equation*}
or equivalently,
\begin{equation*}
  \begin{cases} 
  \left(\frac{u}{\alpha}\right)^{2^s}+\frac{u}{\alpha}+1+\frac{c}{\alpha^{2^s+1}}=0, \\
  c^{2^{2s}}+c^{2^s}\alpha^{2^{3s}-2^s}+c\alpha^{2^{3s}+2^{2s}-2^s-1}+d\alpha^{2^{2s}-1}=0.
  \end{cases}
\end{equation*}
Notice that the first equation of the above system has either $0$ or $4$ solutions in $\F_{2^n}$, depending on whether  $\sum_{i=0}^{t-1} \left(1+\dfrac{c}{\alpha^{2^s+1}}\right)^{2^{si}}$ is zero or not, respectively. This first equation is simply the equation for the DDT entry of $u^{2^s+1}$ at $(\alpha,c)$, and if $\DDT_{u^{2^s+1}}(\alpha,c)=4$, then the four solutions of $\left(\frac{u}{\alpha}\right)^{2^s}+\frac{u}{\alpha}+1+\frac{c}{\alpha^{2^s+1}}=0$ are $\{u, u+\alpha, u+\alpha \omega, u+\alpha \omega^2\}$. This implies that $u^{2^s+1}, (u+\alpha)^{2^s+1}, (u+\alpha \omega)^{2^s+1}, (u+\alpha \omega^2)^{2^s+1}$ are solutions of $F_2(X)+F_2(X+c)=d,$ or equivalently $\{z_1,z_1+c,z_2,z_2+c\} \subseteq \left\{ u^{2^s+1}, (u+\alpha)^{2^s+1}, (u+\alpha \omega)^{2^s+1}, (u+\alpha \omega^2)^{2^s+1}\right\}$. Then for $z_1=u^{2^s+1}$, we have $z_1+c=(u+\alpha)^{2^s+1}$. Further, we split our analysis in the following two cases:

\noindent\textbf{Case 1.} Let $z_2=(u+\alpha\omega)^{2^s+1}$ Then $z_1+z_2=a$ and $F_{2}(z_1)+F_{2}(z_2)=b$ can be written as
$ u^{2^s+1}+(u+\alpha\omega)^{2^s+1}=a$ and
 $ u^{2^{3s}+1}+(u+\alpha\omega)^{2^{3s}+1}=b$.

Also, by using $ u^{2^{3s}+1}+(u+\alpha)^{2^{3s}+1}=d$, and the second equation of above system will give us $b=d\omega+\alpha^{2^{3s}+1}$. 
To compute $a$, we use $ u^{2^s+1}+(u+\alpha)^{2^s+1}=c$ and the first equation of the above system and simplify them to get $a=c \omega+ \alpha^{2^s+1}$.

\noindent\textbf{Case 2.} Now assume $z_2=(u+\alpha\omega^2)^{2^s+1}$. By using the same argument as in the above case, we get $a=c \omega^2+ \alpha^{2^s+1}$ and $b=d \omega^2+ \alpha^{2^{3s}+1}$.

Summarizing both cases, we get the EBCT of the Kasami function.

We next compute the LBCT and UBCT entries via (and using the notations of) Corollary~\ref{E2LU}. As $\DDT_{F_2}(c,d)$ is either $0$ or $4$,   then, $\LB_{F_{2}}(a,c,d) \neq 2$, for any $(a,c,d) \in \F_{2^n}^{*} \times \F_{2^n}^{*} \times \F_{2^n}^{*}$ and, $\UB_{F_{2}}(c,d,b) \neq 2$, for any $(c,d,b) \in \F_{2^n}^{*} \times \F_{2^n}^{*} \times \F_{2^n}^{*}.$ 

Now, one can write set 
\[
A=\left\{  (c,c,d), (\omega c+\alpha^{2^s+1}, c, d), (\omega^2 c +\alpha^{2^s+1}, c, d) \bigm|  \DDT_{F_2}(c,d)=4\right\},
\]
 for $c,d \in \F_{2^n}^{*}$. This is because, when $\DDT_{F_2}(c,d)=4$, $(c, c, d) \in A$, for $b= d$, $(\omega c + \alpha^{2^s+1}, c, d) \in A$, for $b=d \omega+ \alpha^{2^{3s}+1}$ and $(\omega^2 c + \alpha^{2^s+1}, c, d) \in A$, for $b=d \omega^2+ \alpha^{2^{3s}+1}$. Finally, from Corollary~\ref{E2LU}, we have $\LB_{F_2}(a,c,d)=4$, for $a \in \left\{c,\omega c +\alpha^{2^s+1}, \omega^2 c+ \alpha^{2^s+1}\right\}$.  Thus, by reordering the indexes, we complete the LBCT case.

Similarly, one can write the set 
\[
B= \left\{(c,d,d), (c,d,d \omega+ \alpha^{2^{3s}+1}), (c,d,d \omega^2+ \alpha^{2^{3s}+1}) \bigm|  \DDT_{F_2}(c,d)=4\right\},
\]
 as, in the case $\DDT_{F_2}(c,d)=4$, $(c, d, d) \in B$, for $a= c$, $(c, d, d \omega+ \alpha^{2^{3s}+1})\in B$, for $a=\omega c +\alpha^{2^s+1}$ and $(c, d, d \omega^2+ \alpha^{2^{3s}+1}) \in B$, for $a=\omega^2 c +\alpha^{2^s+1}$. Hence, one can get the UBCT of the Kasami function by reordering the indexes.
\end{proof}

\begin{exmp}
 The examples given in Table~\textup{\ref{Table3}} ($g$ denotes a primitive element of $\F_{2^{2t}}$) illustrate Corollary~\textup{\ref{Kasami}}:
\end{exmp}
\begin{table}[hbt]
\caption{EBCT, LBCT and UBCT entries for Kasami function $X^{2^{2s}-2^s+1}$ over $\F_{2^{2t}}$}
\label{Table3}
\begin{center}
\scalebox{0.8}{
\begin{tabular}{|c|c|c|c|c|c|c|c|c|c|c|} 
 \hline
 $t$ & $s$ & $a$ & $b$ & $c$ & $d$ & $\alpha$ & $\omega$ & $\EB_{F_2}(a,b,c,d)$ &$\LB_{F_2}(a,b,c)$ & $\UB_{F_2}(a,b,c)$ \\
 \hline
$5$ & $2$ & $g^{6}$ & $g$ & $g^{605}$  & for any $d$ & $g^{50}$  & $ g^{341}$ & $0$ & $0$ & $4$\\
 \hline
$5$ & $6$ & $g^{5}$ & $g$ & $g^{774}$  & for any $d$ & $g^{994}$  & $ g^{682}$& $0$ & $0$ & $0$\\
 \hline
 $5$ & $2$ & $g^{401}$ & $g^{605}$ & $g^{6}$  & $g$ & $g^{50}$  & $ g^{341}$ & $4$ & $0$ & $0$\\
 \hline
 $5$ & $6$ & $g^{84}$ & $g^{24}$ & $g^{2}$  & for any $d$ & $g^{681}$  & $ g^{682}$ & $0$ & $4$ & $0$\\
 \hline
\end{tabular}
}
\end{center}
\end{table}

Bracken and Leander~\cite{BL} studied the differential uniformity of $F_3(X) = X^{2^{2s}+2^s+1}$ over $\F_{2^n}$ where $n = 4s$, which is often known as the Bracken-Leander function. They showed that $F_3$ is differentially $4$-uniform, and its nonlinearity  is $2^{n-1} -2^{\frac{n}{2}}$.  Further, Calderini and Villa~\cite{CV} investigated the BCT entries and showed that the boomerang uniformity of the Bracken-Leander function is upper bounded by 24. Here, we first recall some results from~\cite{XY} that will be used to compute the EBCT, LBCT and UBCT entries for the Bracken-Leander function.

\begin{lem}\label{bl1}\cite[Lemma 1]{XY}
  For any $X \in \F_{2^n}$ and $c \in \F_{2^n}$, let $t = \Trsn(c+1) \in \F_{2^s}$, $n=4s$.
\begin{enumerate}
 \item[$(1)$] If $F_3(X)+F_3(X+1)=c$, then $ \Trsn(X)=t$.
 \item[$(2)$] Let $Z = X + X^{2^s}$. Then $F_3(x)+F_3(X+1)=c$ holds if and only if  
$  Z^2 + (t + 1)Z + c^{2^s} + c^{2^{2s}}=0$ and
$ X^2 + X + Z^{2^s+1} + Z^{2^s} + Z^2 + c+1 = 0$.
\end{enumerate}
\end{lem}
Next, we define $W$ to be the set of all $X \in \F_{2^n}$ such that there exists $X' \in \F_{2^n} \setminus \{X, X + 1\} $ with $X^{2^{2s}+2^s+1}+(X+1)^{2^{2s}+2^s+1} = X'^{2^{2s}+2^s+1}+(X'+1)^{2^{2s}+2^s+1}$. We now recall a result describing $W$ explicitly.

\begin{lem}\label{bl2}\cite[Theorem 2]{XY}
Given elements $u \in\F_{2^{2s}}\setminus\{\F_{2^s}\}$ and $v\in\F_{2^{n}}\setminus\{\F_{2^{2s}}\}$ such that
$u+u^{2^s}=1, v+v^{2^{2s}}=1,$ 
which always exist, then $X \in W$ if and only if
$X=(tu+\beta)v+\tau,$
where $t \in\F_{2^{s}}\setminus\{1\}, \tau \in\F_{2^{2s}}$, and $\beta \in\F_{2^{s}}$ satisfies
$\beta = (t + 1)^{-1}(\alpha^2 + \alpha)+(t + 1)(u^2 + u) + 1,$
for some $\alpha \in\F_{2^{s}}$.
\end{lem}
\begin{rem}
The proof of the above lemma also shows that $X'=X+(t+1)u+\alpha$. 
\end{rem}
Next, we determine the EBCT, LBCT and UBCT entries for the Bracken-Leander function.

\begin{cor}
\label{Bracken}
Let $F_3(X)=X^{2^{2s}+2^s+1}$ on $\F_{2^n}$, where $n=4s$. Then for $a,b,c,d \in \F_{2^n}$, we have
 \allowdisplaybreaks
{\small \begin{equation*}
\EB_{F_{3}}(a, b, c,d) = 
  \begin{cases} 
   2^n & \mbox{if } a=b=c=d=0,\\ 
  \DDT_{F_3}(a,b) & \mbox{if } c = d =0 \mbox{ and } a \neq 0,\\
  \DDT_{F_3}(c,d) & \mbox{if}~ ac \neq 0, a=c \mbox{ and } b=d; \mbox{ or }a=0, b=0 \mbox{ and } c \neq 0; \\
&\mbox{ or }\DDT_{F_{3}}(c,d)=4,  a= c((t+1)u+\alpha) \mbox{ and }\\
&\qquad  b=d(\alpha+(t+1)u^{2^s}) +F_{3}(c) \gamma;\\
&\mbox{ or }\DDT_{F_{3}}(c,d)=4,  a= c((t+1)u+\alpha+1)\mbox{ and } \\
&\qquad  b=d(\alpha+(t+1)u^{2^s}+1) +F_{3}(c) \gamma,\\
  0 & \mbox{otherwise,}
  \end{cases}
\end{equation*} 
}
where $t=\Trsn\left(\frac{d}{F_{3}(c)}\right)$, and $F_{3}((tu+\beta)v+\tau)+F_{3}((tu+\beta)v+\tau+1)=\frac{d}{F_{3}(c)},$
 \allowdisplaybreaks
{\small \begin{equation*}
\LB_{F_3}(a, b, c) = 
  \begin{cases} 
  2^n & \mbox{if } b=c=0,\\
  \DDT_{F_3}(b,c) & \mbox{if } a=0 \mbox{ and } b \neq 0; \mbox{ or } \\
  & ~\DDT_{F_3}(b,c)=2, ab \neq 0 \mbox{ and }  a=b; \mbox{ or } \\
& ~\DDT_{F_3}(b,c)=4, ab \neq 0\mbox{ and} \\
& \qquad a\in \{b, b((t+1)u+\alpha),b((t+1)u+\alpha+1)\},  \\
  0 & \mbox{otherwise,}
  \end{cases}
\end{equation*} 
}
where $t=\Trsn\left(\frac{c}{F_{3}(b)}\right)$, and $F_{3}((tu+\beta)v+\tau)+F_{3}((tu+\beta)v+\tau+1)=\frac{c}{F_{3}(b)},$ and,
 \allowdisplaybreaks
{\small \begin{equation*}\UB_{F_3}(a, b, c) = 
  \begin{cases} 
  |F_{3}^{-1}(c+\Im(F_3))|& \mbox{if}~a= b= 0, \\
  \DDT_{F_3}(a,b) & \mbox{if}~ \DDT_{F_3}(a,b)=2, a \neq 0 \mbox{ and }  c=b; \mbox{ or }\\
  &~\DDT_{F_3}(a,b)=4, a \neq 0, \\
  &~ c\in \{ b,b(\alpha+(t+1)u^{2^s}) +F_{3}(a) \gamma,  \mbox{ and} \\
   &\qquad b(\alpha+(t+1)(1+u)+1) + F_{3}(a) \gamma\},\\
  0 & \mbox{otherwise,}
  \end{cases}
\end{equation*} 
}
where $t=\Trsn\left(\frac{b}{F_{3}(a)}\right)$, $F_{3}((tu+\beta)v+\tau)+F_{3}((tu+\beta)v+\tau+1)=\frac{b}{F_{3}(a)},$ $\Im(F_3)$ is the image of $F_3$ and $F_{3}^{-1}(\cdot)$ denotes the preimage of the argument. In particular, if $F_3$ is a permutation, or equivalently if $s$ is odd, then $|F_{3}^{-1}(c+\Im(F_3))|=2^n$.

In all the computed tables, $\gamma= ((t+1)^{-1}(\alpha^4+\alpha^2)+(t+1)^3(u^4+u^2)+\alpha^2t+\alpha t+t(t+1)(u+u^2))$, $\tau \in \F_{2^{2s}}$, $u \in\F_{2^{2s}}\setminus\{\F_{2^s}\}$, $v\in\F_{2^{n}}\setminus\{\F_{2^{2s}}\}$ satisfying $u + u^{2^s} = 1$, $v + v^{2^{2s}}= 1,$ and $\beta \in\F_{2^{s}}$ satisfying $\beta = (t+1)^{-1}(\alpha^2 + \alpha)+(t+1)(u^2 + u) + 1$, for some $\alpha \in\F_{2^{s}}$.

\end{cor}
\begin{proof}
 By Corollary ~\ref{4diff}, $\EB_{F_{3}}(a, b, c,d) $ is equal to
{\small \begin{align*} 
 \begin{cases} 
\DDT_{F_3}(a,b) & \mbox{if } c = d =0,\\
  \DDT_{F_3}(c,d) & \mbox{if}~ ac \neq 0, a=c \mbox{ and } b=d; \mbox{ or }a=0, b=0, c \neq 0; \mbox{ or }\\
&~ \DDT_{F_3}(c,d)=4, a=z_1+z_2 \mbox{ and }  b=F_3(z_1)+F_3(z_2);\mbox{ or }\\
&~ \DDT_{F_3}(c,d)=4,a=z_1+z_2+c \mbox{ and }b=F_3(z_2)+F_3(z_2+c),\\
  0 & \mbox{otherwise},
  \end{cases}
\end{align*} 
}where $z_1,z_1+c,z_2,z_2+c$ are the four solutions of the equation $F_{3}(X+c)+F_{3}(X)=d$ if $\DDT_{F_{3}}(c,d)=4$. This is equivalent to $F_{3}\left(\frac{X}{c}+1\right)+F_{3}\left(\frac{X}{c}\right)=\frac{d}{c^{2^{2s}+2^s+1}}$ having solutions $\left\{\frac{z_1}{c},\frac{z_1}{c}+1,\frac{z_2}{c},\frac{z_2}{c}+1\right\}$. We can simplify $F_{3}\left(\frac{X}{c}+1)+F_{3}(\frac{X}{c}\right)=\frac{d}{F_{3}(c)}$ as
{\small
$$\left(\frac{X}{c}\right)^{2^{2s}+2^s}+\left(\frac{X}{c}\right)^{2^{2s}+1}+\left(\frac{X}{c}\right)^{2^s+1}+\left(\frac{X}{c}\right)^{2^{2s}}+\left(\frac{X}{c}\right)^{2^s}+\frac{X}{c}+1=\frac{d}{F_{3}(c)}.$$
}
Let $\Trsn\left(\frac{d}{F_{3}(c)}+1\right)=\Trsn\left(\frac{d}{F_{3}(c)}\right)=t$, where $\Trsn$ denotes the relative trace map from $\F_{2^n}$ to~$\F_{2^s}$. 

\begin{sloppypar}
First assume that $t=1$, then from the proof of \cite[Theorem 1]{BL}, we know that $F_{3}(\frac{X}{c}+1)+F_{3}(\frac{X}{c})=\frac{d}{F_{3}(c)}$ can have at most two solutions, a contradiction. Next, for $t\neq 1$, from Lemma~\ref{bl2}, for $u \in\F_{2^{2s}}\setminus\{\F_{2^s}\}$ and $v\in\F_{2^{n}}\setminus\{\F_{2^{2s}}\}$ such that $u + u^{2^s} = 1$, $v + v^{2^{2s}}= 1$, we know that $\frac{z_1}{c} \in W$ if and only if $\frac{z_1}{c}=(tu+\beta)v+\tau$ and $\frac{z_2}{c} \in \left\{\frac{z_1}{c}+ (t+1)u+\alpha, \frac{z_1}{c}+(t+1)u+\alpha+1\right\}$, where $\tau \in\F_{2^{2s}}$, and $\beta \in\F_{2^{s}}$ satisfies the relation $\beta = (t+1)^{-1}(\alpha^2 + \alpha)+(t+1)(u^2 + u) + 1$, for some $\alpha \in\F_{2^{s}}$. Thus, $F_{3}(\frac{X}{c}+1)+F_{3}(\frac{X}{c})=\frac{d}{F_{3}(c)}$ can be written as $F_{3}((tu+\beta )v+\tau)+F_{3}((tu+\beta) v+\tau+1)=\frac{d}{F_{3}(c)}$.
\end{sloppypar}

Again, we divide our analysis into the following two cases:

\noindent\textbf{Case 1.} Let $\frac{z_2}{c} = \frac{z_1}{c}+ (t+1)u+\alpha$, then we have $a=c((t+1)u+\alpha)$. Now, $b=F_{3}(z_1)+F_{3}(z_2)=z_{1}^{2^{2s}+2^s+1}+z_{2}^{2^{2s}+2^s+1}=c^{2^{2s}+2^s+1}(F_{3}(\frac{z_1}{c})+F_{3}( \frac{z_1}{c}+ (t+1)u+\alpha))$. For $y_1=\frac{z_1}{c}$ and $D=\frac{d}{F_{3}(c)}$, we have
\allowdisplaybreaks
\begin{align*}
 \frac{b}{F_{3}(c)} &= F_{3}\left(y_1\right)+F_{3}\left( y_1+ (t+1)u+\alpha \right) \\
  & = (\alpha+(t+1)u)(y_{1}^{2^{2s}+2^s}+y_{1}^{2^{2s}+1}+y_{1}^{2^s+1}) + (\alpha+(t+1)u)^2(y_{1}^{2^{2s}}+y_{1}^{2^k}+y_{1})\\
  &\qquad  + (t+1)(y_{1}^{2^{2s}+1}+(\alpha+(t+1)u)(y_{1}+y_{1}^{2^s}))+(\alpha+(t+1)u)^{2^{2s}+2^s+1}\\
  & = ((\alpha+(t+1)u) + (\alpha+(t+1)u)^2)(y_{1}^{2^{2s}}+y_{1}^{2^s}+y_{1})+(\alpha+(t+1)u)(1+D)\\
  &\qquad  + (t+1)(y_{1}^{2^{2s}+1}+(\alpha+(t+1)u)(y_{1}+y_{1}^{2^s}))+(\alpha+(t+1)u)^{2^{2s}+2^s+1}\\
  & = (t+1)(tu+\beta+1)(y_{1}^{2^{2s}}+y_{1}^{2^s}+y_{1})+(\alpha+(t+1)u)(1+D)\\
  &\qquad  + (t+1)(y_{1}^{2^{2s}+1}+(\alpha+(t+1)u)(y_{1}+y_{1}^{2^s}))+(\alpha+(t+1)u)^{2^{2s}+2^s+1}\\
  & = (t+1)(tu+\beta+1)(tu+\beta+(tu^{2^s}+\beta)v^{2s}+\tau^{2^s})+(\alpha+(t+1)u)(1+D)\\
  & \qquad +(t+1)((tu+\beta)^2v^{2^{2s}+1}+\tau(tu+\beta)+\tau^2\\
  &\qquad +(\alpha+(t+1)u)(tu+\beta))+(\alpha+(t+1)u)^{2^{2s}+2^s+1}\\
  & = (t+1)(D+tu+\beta+1)+(t+1)(tu+\beta+1)(tu+\beta)+ (\alpha+(t+1)u)(1+D)\\
  &\qquad  +(t+1)(\alpha+(t+1)u)(tu+\beta)+(\alpha+(t+1)u)^{2^{2s}+2^s+1} \\
  & = (\alpha+(t+1)u^{2^s})D+(t+1)^{-1}(\alpha^4+\alpha^2)+(t+1)^3(u^4+u^2)\\
  &\qquad +\alpha^2t+\alpha t+t(t+1)(u+u^2).
\end{align*}
The above expression will give us $b=d(\alpha+(t+1)u^{2^s})+F_{3}(c)((t+1)^{-1}(\alpha^4+\alpha^2)+(t+1)^3(u^4+u^2)+\alpha^2t+\alpha t+t(t+1)(u+u^2))=d(\alpha+(t+1)u^{2^s})+F_{3}(c)\gamma$ (say), where $\gamma= ((t+1)^{-1}(\alpha^4+\alpha^2)+(t+1)^3(u^4+u^2)+\alpha^2t+\alpha t+t(t+1)(u+u^2))$.

\noindent\textbf{Case 2.} Let $\frac{z_2}{c} = \frac{z_1}{c}+ (t+1)u+\alpha+1$, then by the same argument as in Case 1, we have $a=c((t+1)u+\alpha+1)$ and  $b=d(\alpha+(t+1)u^{2^s}+1)+F_{3}(c)((t+1)^{-1}(\alpha^4+\alpha^2)+(t+1)^3(u^4+u^2)+\alpha^2t+\alpha t+t(t+1)(u+u^2)).$

Combining both of these cases, we get the EBCT entries for the Bracken-Leander function.

Next, we compute the UBCT and LBCT entries for the Bracken-Leander function. Using Corollary~\ref{4diff} the LBCT entries are straightforward to determine when $abc=0$, except the case where $a \neq b, ab \neq 0$ and $c=0$. Furthermore, if $ab \neq 0, a \neq b$ and $c=0$, the conditions derived for LBCT entries are identical to those obtained for the  case where $abc \neq 0$ and $a \neq b$. Similarly, in the case of UBCT entries, we only need to focus on the case when $a \neq 0$ and $bc=0$ as the other cases are trivial if $abc=0$. The conditions in this case are identical to those for $abc \neq 0$. Therefore, it is sufficient to determine the UBCT and LBCT entries when $a,b,c,d \in \F_{2^n}^*$.

We then use (also its notations) of Corollary~\ref{E2LU} to compute the LBCT and UBCT entries for the Bracken-Leander function when $abcd \neq 0$. As we have computed $\EB_{F_3}(a,b,c,d)=2$, when $a=c, b=d$ and $\DDT_{F_3}(c,d) =2$, then $\LB_{F_{3}}(a,c,d) = 2$, for $a=c$ when $\DDT_{F_3}(c,d) =2$ and, $\UB_{F_{3}}(c,d,b) = 2$, for $b=d$ when $\DDT_{F_3}(c,d) =2$.

Also, $A=\{  (c,c,d), (c((t+1)u+\alpha), c, d), (c((t+1)u+\alpha+1), c, d) \bigm|  \DDT_{F_3}(c,d)=4\}$, as, in the case $\DDT_{F_3}(c,d)=4$, $(c, c, d) \in A$, for $b= d$, $(c((t+1)u+\alpha), c, d) \in A$, for $b=d(\alpha+(t+1)u^{2^s})+F_{3}(c)\gamma$ and $(c((t+1)u+\alpha+1), c, d) \in A$, for $b=d(\alpha+(t+1)u^{2^s}+1)+F_{3}(c)\gamma$. Thus, from Corollary~\ref{E2LU}, we have $\LB_{F_3}(a,c,d)=4$, for $a \in \left\{c,c((t+1)u+\alpha), c((t+1)u+\alpha+1)\right\}$. This completes the LBCT case.

Similarly, one can write $B= \{(c,d,d), (c,d,d(\alpha+(t+1)u^{2^s})+F_{3}(c)\gamma), (c,d,d(\alpha+(t+1)u^{2^s}+1)+F_{3}(c)\gamma\bigm|  \DDT_{F_3}(c,d)=4\}$, as, in the case of $\DDT_{F_3}(c,d)=4$, $(c, d, d) \in B$, for $a= c$, $(c, d, d(\alpha+(t+1)u^{2^s})+F_{3}(c)\gamma)\in B$, for $a=c((t+1)u+\alpha)$ and $(c, d, d(\alpha+(t+1)u^{2^k}+1)+F_{3}(c)\gamma) \in B$, for $a=c((t+1)u+\alpha+1)$.

Hence, we get the LBCT and UBCT of the Bracken-Leander function by just reordering the indexes.
\end{proof}

\begin{exmp}
 To illustrate Corollary~\textup{\ref{Bracken}}, Table~\textup{\ref{Table4}} gives explicit values of the LBCT and UBCT entries of Bracken-Leander function over $\F_{2^{4s}}^*=\langle g\rangle$ ($g$ is a primitive element of  $\F_{2^{4s}}$).
\end{exmp}
\begin{table}[hbt]
\caption{EBCT, LBCT and UBCT entries for Bracken-Leander function $X^{2^{2s}+2^s+1}$ over $\F_{2^{4s}}^{*}=\langle g\rangle$}
\label{Table4}
\begin{center}
\scalebox{0.8}{
\begin{tabular}{|c|c|c|c|c|c|c|c|c|c|c|c|} 
 \hline
  $s$ & $a$ & $b$ & $c$ & $d$ & $\alpha$ & $u$ & $v$ & $\tau$ & $\EB_{F_3}(a,b,c,d)$ & $\LB_{F_3}(a,b,c)$ & $\UB_{F_3}(a,b,c)$\\
  \hline
  $2$ & $g$ & $g$ & $g$ & $d$ & $-$ & $-$ & $-$ & $-$ & $2$ & $2$ & $2$\\
 \hline
  $2$ & $g^{71}$ & $g^3$ & $g^{32}$ & $d$ & $g^{255}$ & $g^{17}$ & $g^{123}$ & $g^{17}$ & $0$ & $4$ & $0$\\
 \hline
  $2$ & $g^{54}$ & $g^{37}$ & $g^{26}$ & $d$ & $g^{85}$ & $g^{34}$ & $g^{224}$ & $g^{17}$ & $0$ & $4$ & $0$\\
 \hline
 $2$ & $g^{70}$ & $g^{3}$ & $g^{103}$& $d$ & $g^{85}$ & $g^{34}$ & $g^{224}$ & $g^{17}$ & $0$ & $0$ & $4$\\
 \hline
  $2$ & $ g^{36}$ & $g^{103}$ & $g^{70}$& $g^3$ & $g^{85}$ & $g^{34}$ & $g^{224}$ & $g^{17}$ & $4$ & $0$ & $0$\\
 \hline
\end{tabular}
}
\end{center}
\end{table}
Via our Corollary \ref{4diff},  we can also give a very short proof for the EBCT, LBCT and UBCT entries of the inverse function for $n$ even, which is the main result in \cite{EM, MLLZ} (for $n$ odd, see the previous section). Since $F_4$ is a permutation, the case for $abcd=0$ follows directly from Lemma~\ref{L0001}. Therefore, we focus our discussion in the following corollary on the cases where $a,b,c,d \in \F_{2^n}^{*}$.

\begin{cor}
\label{Inverse}
 Let $F_4(X)=X^{2^n-2}$ on $\F_{2^n}$, where $n$ is even. For $a,b,c,d \in \F_{2^n}^{*}$, we have
  \allowdisplaybreaks
 {\small  \begin{align*}
 \EB_{F_4}(a, b, c,d) &= 
  \begin{cases} 
4&\mbox{if }~a=c=\frac{1}{b}=\frac{1}{d}; \mbox{ or }\\
  &~ b=\frac{1}{a}, d=\frac{1}{c}, \mbox{ and } (ad)^2+ad+1=0,\\ 
2& \mbox{if}~a=c,b=d, d\neq\frac{1}{c}\mbox{ and } \Tr\left(\frac{1}{cd}\right)=0,\\
  0 & \mbox{otherwise,}
  \end{cases}\\
\LB_{F_4}(a, b, c) &= 
  \begin{cases} 
4&\mbox{if }~a=b=\frac{1}{c}; \mbox{ or }\\
  &~ b=\frac{1}{c},(ac)^2+ac+1=0,\\
2& \mbox{if}~a=b\neq\frac{1}{c}, \Tr\left(\frac{1}{bc}\right)=0,\\
  0 & \mbox{otherwise,}
  \end{cases}\\
\UB_{F_4}(a, b, c) &= 
  \begin{cases} 
4&\mbox{if }~b=c=\frac{1}{a};\mbox{ or }\\
  &~ b=\frac{1}{a},(ac)^2+ac+1=0,\\
2& \mbox{if}~c=b\neq\frac{1}{a}, \Tr\left(\frac{1}{ab}\right)=0,\\
  0 & \mbox{otherwise.}~
  \end{cases}
\end{align*}  
}
\end{cor}
\begin{proof}
For all three results, we have to study the spectrum of the inverse function, which can be found, for instance, implied in the proofs in~\cite{EM, MLLZ}, but we recall it below for easy reference.

The equation $(X+A)^{^{2^n-2}}+X{^{2^n-2}}=B$, for $A,B\neq0$, and $n$ even, gives:
\begin{itemize}
\item If  $X=0$ or $X=A$, then  $B=\frac{1}{A}$, otherwise there are no solutions.
\item If $X\neq0,A$, then the equation is equivalent to $X^2+AX+\frac{A}{B}=0$, which has two solutions if and only if $\Tr\left(\frac{1}{AB}\right)=0$, and zero solutions, otherwise.
\end{itemize}

By Corollary \ref{4diff}, $\EB_{F_4}(a, b, c,d) $ is equal to
{\small 
\begin{equation*} 
  \begin{cases} 
  \DDT_{F_4}(c,d) & \mbox{ if }  a=c,b=d ,\mbox{ or}\\
&\quad \DDT_{F_4}(c,d)=4, a=z_1+z_2,  b=F_4(z_1)+F_4(z_2),\mbox{ or}\\
&\quad \DDT_{F_4}(c,d)=4, a=z_1+z_2+c, b=F_4(z_1)+F_4(z_2+c)\\
  0 & \mbox{otherwise,}
  \end{cases}
\end{equation*} 
}
where $z_1,z_1+c,z_2,z_2+c$ are the four solutions of the equation $F_4(X+c)+F_4(X)=d$, if $\DDT_{F_4}(c,d)=4$. 

If $a=c,b=d$, $d\neq\frac{1}{c}$ and $\Tr\left(\frac{1}{cd}\right)=0$, then $\EB_{F_4}(a, b, c,d) =\DDT_{F_4}(c,d)=2$. 
If $d=\frac{1}{c}$, and without loss of generality, we let  $z_1=0,z_2=\omega c$. If $a=c,b=d$, then   $\EB_{F_4}(a, b, c,d) =\DDT_{F_4}(c,d)=4$. If $a=z_1+z_2=\omega c$ and $b=F_4(z_1)+F_4(z_2)=\frac{1}{\omega c}$, then $\EB_{F_4}(a, b, c,d) =\DDT_{F_4}(c,d)=4$. If $a=z_1+z_2+c=\omega^2 c$ and $ b=F_4(z_1)+F_4(z_2+c)=\frac{1}{\omega^2 c}$, then   $\EB_{F_4}(a, b, c,d) =\DDT_{F_4}(c,d)=4$. Note that these last two cases can be summarized as $b=\frac{1}{a}$, $d=\frac{1}{c}$ and $(ad)^2+ad+1=0$. In all other cases, $\EB_{F_4}(a, b, c,d) =0$.

We now compute the LBCT and UBCT of the inverse function using Corollary~\ref{E2LU}. As we have seen $\EB_{F_4}(a, b, c,d)=2$, when $a=c,b=d$, $d\neq\frac{1}{c}$ and $\Tr\left(\frac{1}{cd}\right)=0$, we infer that $\LB_{F_4}(a,c,d)=2$ if $a=c \neq \frac{1}{d}$ and $\Tr\left(\frac{1}{cd}\right)=0$. Similarly, $\UB_{F_4}(c,d,b)=2$ if $b=d \neq \frac{1}{d}$ and $\Tr\left(\frac{1}{cd}\right)=0$.

Clearly, when $d=\frac{1}{c}$, we can write $A=\{  (c,c,d), (\omega c, c, d), (\omega^2 c, c, d)\bigm|  \DDT_{F_4}(c,d)=4\}$, as, in the case $\DDT_{F_4}(c,d)=4$, $(c, c, d) \in A$, for $b= d$, $(\omega c, c, d) \in A$, for $b=\frac{1}{\omega c}$ and $(\omega^2 c, c, d) \in A$, for $b=\frac{1}{\omega^2 c}$. Thus, from Corollary~\ref{E2LU}, we have $\LB_{F_4}(a,c,d)=4$, for $a \in \{c,\omega c, \omega^2 c\}$ and $d=\frac{1}{c}$. This is equivalent to  $\LB_{F_4}(a,c,d)=4$ when $d=\frac{1}{c}$ and $(ad)^2+ad=1=0$.

In a similar way, one can write $B= \{(c,d,d), (c,d,\frac{1}{\omega c}), (c,d,b=\frac{1}{\omega c})\bigm|  \DDT_{F_4}(c,d)=4\}$, as, in the case $\DDT_{F_4}(c,d)=4$, $(c, d, d) \in B$, for $a= c$, $(c, d, \frac{1}{\omega c})\in B$, for $a=\omega c$ and $(c, d, \frac{1}{\omega c}) \in B$, for $a=\omega^2 c$, and the claim follows.
\end{proof}
Using Remark~\ref{r1} and our results for the inverse function, we can easily compute its FBCT (also found in~\cite{EM22}).

\begin{cor} Let $F_4(X)=X^{2^n-2}$ over $\F_{2^n}$, with $n$ even. Then, for $a,b\in\F_{2^n}^*$,
{\small
\begin{equation*}
\FB_{F_4}(a, b) = 
  \begin{cases} 
2^n, & \mbox{if } a=b,\\
4,& \mbox{if } \left(\frac{a}{b}\right)^2+\frac{a}{b}+1=0,\\ 
  0 & \mbox{otherwise.}
  \end{cases}
\end{equation*} 
}
\end{cor}
\begin{proof} If $a=b$, then trivially $\FB_{F_4}(a,b)=2^n$ for any function $F_4$. If $a\neq b$, by Corollary \ref{Inverse}, $\LB_{F_4}(a,b,c)$ is nonzero if and only if $ b=\frac{1}{c},(ac)^2+ac+1=0$, which is equivalent to $b=\frac{1}{c},\left(\frac{a}{b}\right)^2+\frac{a}{b}+a=0$.  In that case, $\LB_{F_4}(a,b,c)=4$. The result follows.
\end{proof}

\section{DBCT for the Gold function}\label{oldS7}

In this section, we study the DBCT entries for the Gold function. To this end, we need the LBCT and UBCT entries of the Gold function, which we have already computed in Corollary~\ref{ggold}.
Also, if $s|n$, then for any $\alpha\in\F_{2^n}$, we have $\sum_{i=0}^{t-1}\alpha^{2^{si}}=\Tr_s^n(\alpha)$, the relative trace of $\F_{2^n}$ over $\F_{2^s}$. If $s\nmid n$, then the elements of $\F_{2^n}$ can be embedded in $\F_{2^{sm}}$, since $m=\frac{n}{\gcd(s,n)}$.

\begin{thm}\label{dbct}
 Let $F_1(X)=X^{2^s+1}$ over $\F_{2^n}$ where $\gcd(s,n)=t$ and $m=\frac{n}{t}$ is odd. Then,
\begin{equation*}\DB_{F_1}(a,d)=
 \begin{cases}
2^{2n}& \mbox{~if~}ad=0,\\
  2^{2t}\left( (2^{t}-2)+|  N(a,d)| \right) & \mbox{~if~} ad\neq0,\Tr_s^{sm}\left(\frac{d}{ (a^{2^s+1})^{2^s+1}}\right)= u^4, u \in \F_{2^{t}}^{*} \setminus \{1\}, \\
2^{2t} \cdot |  N(a,d) |  & \mbox{otherwise,}
 \end{cases}
\end{equation*}
where $|N(a,d)|$ denotes the cardinality  of $\left\{b \in \F_{2^n}^{*}\bigm|  \Tr_s^{sm}\left(\frac{b}{a^{2^s+1}}\right)=\Tr_s^{sm}\left(\frac{d}{b^{2^s+1}}\right)=m\right\}$.
\end{thm}
\begin{proof}
First, assume that $a=0$ or $d=0$, then using definition of DBCT, we can write $\DB_{F_1}(a,d)=2^{2n}$. Here $F_1$ is a permutation as $m$ is given to be odd. Thus, for $a,d \in \F_{2^n}^{*}$, we can write
 $$ \DB_{F_1} (a, d) = \sum_{b=c}\UB_{F_1}(a, b, c)  \LB_{F_1} (b, c, d) + \sum_{b \neq c}\UB_{F_1}(a, b, c)  \LB_{F_1} (b, c, d).$$
 Recall, for $a,d \in \F_{2^n}^{*}$, $t=\gcd(k,n)$ and $m=\frac{n}{t}$,
 \allowdisplaybreaks
\begin{align*}
\LB_{F_1}(b, c, d) &= 
\begin{cases}
2^{t}&\mbox{if } \Tr_s^{sm}\left(\frac{d}{c^{2^s+1}}\right)=\Tr_{1}^{m}(1)  \mbox{ and }  b\in c\F_{2^{t}}^*  \mbox{ for }b,c,d \in \F_{2^n}^{*}, \\
\DDT_{F_1}(c,d) & \mbox{if } b=0,\\
  0 & \mbox{otherwise,}
\end{cases}\\
\UB_{F_1}(a, b, c) &= 
\begin{cases}
2^{t}&\mbox{if }\Tr_s^{sm}\left(\frac{b}{a^{2^s+1}}\right)=\Tr_{1}^{m}(1) \mbox{ and~for~} a,b,c \in \F_{2^n}^{*}, \\
& \qquad c\in \left\{b, (u+u^2)a^{2^s+1}+u b, u\in\F_{2^{t}}^*\setminus \{1\}\right\}, \\
\DDT_{F_1}(a,b) & \mbox{if } c=0,\\
  0 & \mbox{otherwise.}
\end{cases}
\end{align*}
Clearly for $b=0$, we have $\DDT_{F_1}(a,b)=\DDT_{F_1}(a,0)=0$ as $a \neq 0$, thus, giving us $\UB_{F_1}(a,b,c)=0$ and similarly for $c=0$, we get $\DDT_{F_1}(c,d)=\DDT_{F_1}(0,d)=0$ as $d \neq 0$ giving us $\LB_{F_1}(a,b,c)=0$.

 When $b\neq c$, then $\UB_{F_1}(a,b,c)=2^{t}$ only if $c=(u+u^2)a^{2^s+1}+u b$ and $\DDT_{F_1}(a,b)=2^{t}$, where $ u\in\F_{2^{t}}^*\setminus \{1\}$. Also, $\LB_{F_1}(b,c,d)=2^{t}$ only if $b=vc$ for some $v\in\F_{2^l}^*$ and $\DDT_{F_1}(c,d)=2^{t}$. 
 After combining them, we have $c=(u+u^2)a^{2^s+1}+uvc$, and $\DDT_{F_1}(a,b)=\DDT_{F_1}(c,d)=2^{t}$. Notice that $v \neq u^{-1}$, because otherwise $a=0$. Thus, we have $c=\frac{(u+u^2)}{1+uv}a^{2^s+1}$ and therefore, $b=\frac{v(u+u^2)}{1+uv}a^{2^s+1}$. Then $\UB_{F_1}(a,b,c)=\LB_{F_1}(b,c,d)=2^{t}$ if $\DDT_{F_1}(a,b)=\DDT_{F_1}(c,d)=2^{t}$ which is the same as $\sum_{i=0}^{m-1}\frac{b^{2^{si}}}{a^{2^{si}(2^s+1)}}=\sum_{i=0}^{m-1}\frac{d^{2^{si}}}{c^{2^{si}(2^s+1)}}=m$, where $m=\frac{n}{t}$. This is equivalent to $\sum_{i=0}^{m-1} \left(\frac{v(u+u^2)}{1+uv}\right)^{2^{si}} =\sum_{i=0}^{m-1}\frac{d^{2^{si}}}{(\frac{(u+u^2)}{1+uv} a^{2^s+1})^{2^{si}(2^s+1)}}=m$. One can see that the first equality 
 $\sum_{i=0}^{m-1} \left(\frac{v(u+u^2)}{1+uv}\right)^{2^{si}}=m$, is nothing but $m \left(\frac{v(u+u^2)}{1+uv}\right)=m$, as $u,v \in \F_{2^{t}} \subseteq \F_{2^s}$. This is the same as saying $m\dfrac{u^2v+1}{uv+1}=0$, or equivalently we have $\frac{u^2v+1}{uv+1}=0$ or equivalently, $v=\frac{1}{u^2}$, as $m$ is odd. Then we have $\frac{u^2v+1}{uv+1}=0$ or equivalently, $v=\frac{1}{u^2}$. This will give us $c=u^2(a^{2^s+1})$ and $b=a^{2^s+1}$. Thus, for $b=a^{2^s+1}$ and $c=u^2(a^{2^s+1})$, we have $\sum_{i=0}^{m-1} \frac{b^{2^{si}}}{a^{2^{si}(2^s+1)}}=1 (\equiv m$ (mod 2)). Also, $\sum_{i=0}^{m-1}\frac{d^{2^{si}}}{( u^2 a^{2^s+1})^{2^{si}(2^s+1)}}=m$ for some $a,d \in \F_{2^{t}}$. Obviously, such $a$ and $d$ always exist because $\sum_{i=0}^{m-1}\frac{d^{2^{si}}}{( u^2 a^{2^s+1})^{2^{si}(2^s+1)}}=1$ if and only if there exists an $A\in\F_{{2}^{sm}}$ such that $\frac{d^{2^{si}}}{( u^2 a^{2^s+1})^{2^{si}(2^s+1)}}=1+A+A^{2^s}$. Finally, we can conclude that for all $u \in \F_{2^{t}}$,  choosing $c=u^2(a^{2^s+1})$ and $b=a^{2^s+1}$, for $\UB_{F_1}(a,b,c)=\LB_{F_1}(b,c,d)=2^{t}$, we should have $\sum_{i=0}^{m-1}\frac{d^{2^{si}}}{u^4( a^{2^s+1})^{2^{si}(2^s+1)}}=m$. Therefore, 
 $\DB_{F_1}(a,d)= 
  2^{2t} (2^l-2)  \mbox{~if~} \Tr_s^{sm}\left(\frac{d}{  (a^{2^s+1})^{2^s+1}}\right)= u^4$, for some $u \in \F_{2^{t}}^{*} \setminus \{1\}$.

If $b=c$, then
$
\sum_{b = c}\UB_{F_1}(a, b, c)  \LB_{F_1} (b, c, d) = 
\sum_{b} \DDT_{F_1}(a,b) \DDT_{F_1}(b,d),
$
when  $\sum_{i=0}^{m-1}\frac{b^{2^{si}}}{a^{2^{si}(2^s+1)}}=\sum_{i=0}^{m-1}\frac{d^{2^{si}}}{b^{2^{si}(2^s+1)}}=m$. If any of these sums is different from $m$, then the product of the DDT entries  is~0.
 Hence, we need to compute $|N(a,d)|$, that is, 
 \allowdisplaybreaks
 \begin{align*}
N(a,d)&=\left\{b \in \F_{2^n}^{*}\bigg| \sum_{i=0}^{m-1}\frac{b^{2^{si}}}{a^{2^{si}(2^s+1)}}=\sum_{i=0}^{m-1}\frac{d^{2^{si}}}{b^{2^{si}(2^s+1)}}=m\right\}\\
&=\left\{b \in \F_{2^n}^{*}\bigm|  \Tr_s^{sm}\left(\frac{b}{a^{2^s+1}}\right)=\Tr_s^{sm}\left(\frac{d}{b^{2^s+1}}\right)=m\right\},
\end{align*}
which completes the proof.
\end{proof}

\begin{exmp}
 To illustrate Theorem~\textup{\ref{dbct}}, in Table~\textup{\ref{Table5}} we give explicit values of the DBCT entries of the Gold function over $\F_{2^n}^*=\langle g\rangle$, where $g$ is a primitive element of  $\F_{2^n}$.
\end{exmp}
\begin{table}[hbt]
\caption{DBCT entries for the Gold function $X^{2^s+1}$ over $\F_{2^{n}}$}
\label{Table5}
\begin{center}
\begin{tabular}{|c|c|c|c|c|c|}
 \hline
 $n$ & $s$ & $a$ & $d$ &  $u$ & $\DB_{F_1}(a,d)$ \\
 \hline
 $6$ & $2$ & $g^{25}$ & $g^{22}$ &  $g^3 + g^2 + g + 1$ & $160$ \\
 \hline
  $6$ & $2$ & $g^{63}$ & $g^{56}$ &  -- & $64$ \\
 \hline
  $10$ & $2$ & $g^{4}$ & $g^{186}$ &  -- & $1024$ \\
 \hline
 $10$ & $4$ & $g^{2}$ & $g^{868}$ &  $g^5 + g^3 + g + 1$ & $1024$ \\
 \hline
\end{tabular}
\end{center}
\end{table}

\section{Conclusion}\label{S9}

We have determined the entries of the newly proposed Extended Boomerang Connectivity Table (EBCT), Lower Boomerang Connectivity Table (LBCT) and Upper Boomerang Connectivity Table (UBCT) of a $\delta$-uniform function by establishing connections with the entries of the well-known Difference Distribution Table (DDT), which we consider to be a significant observation. We also give the EBCT, LBCT and UBCT entries of three classes of differentially 4-uniform power permutations, namely, Gold, Kasami and Bracken-Leander. Moreover, it turns out that one of our results covers the result of Eddahmani et al.~\cite{EM} and Man et al.~\cite{MLLZ} for the inverse function over $\F_{2^n}$. We also compute the DBCT entries for the Gold function. We further challenge the community to determine the DBCT entries of other classes of interesting functions, such as Kasami and Bracken-Leander functions over $\F_{2^n}$.

\section*{Acknowledgements}
The authors would like to thank the editor, Prof. Daniele Bartoli, for the prompt handling of our paper, and  they extend their appreciation  to the very professional referees, who spotted errors/typos (now fixed), and provided beneficial and constructive comments to improve our paper. In particular, the referees challenged us to describe the CCZ/EA/A-equivalence of the considered concepts, which is now part of Section~\ref{oldS8}, or provide examples, whenever possible.

\end{document}